\documentclass[12pt]{article} 
\usepackage[table,dvipsnames]{xcolor}
\usepackage[colorlinks, citecolor={blue}]{hyperref}
\usepackage{url}
\usepackage{amsfonts,amscd,amssymb}
\usepackage{amsthm,amsmath,natbib}
\usepackage{algorithm,algorithmicx,algpseudocode}
\usepackage{bm}
\usepackage{bbm} 
\let\oldtabular\tabular
\let\endoldtabular\endtabular
\renewenvironment{tabular}{\rowcolors{3}{trevorblue!15}{white}\oldtabular}{\endoldtabular}
\usepackage{verbatim}
\usepackage{graphicx}
\usepackage{setspace}
\usepackage{natbib}
\usepackage[margin=1in]{geometry}
\usepackage{enumitem}
\usepackage{listings}
\usepackage[textsize=tiny]{todonotes}
\usepackage{tikz}
\usepackage{etoolbox}
\usepackage{appendix}
\usepackage[format=plain,
font=it]{caption}
\DeclareCaptionLabelFormat{andtable}{#1~#2  and  \tablename~\thetable}
\usepackage{subcaption}
\usepackage{wrapfig}
\usepackage{xr}
\usepackage{booktabs}
\usepackage{multirow}
\usepackage{authblk}
\usepackage{adjustbox}
\usepackage{floatpag}
\usepackage{fancyvrb}

\makeatletter
\patchcmd{\@verbatim}
{\verbatim@font}
{\verbatim@font\scriptsize}
{}{}
\makeatother
\usepackage{multicol}

\usetikzlibrary{matrix}
\usetikzlibrary{backgrounds}
\usetikzlibrary{calc}
\usetikzlibrary{arrows,shapes}
\usetikzlibrary{decorations}
\usetikzlibrary{decorations.pathmorphing}
\usetikzlibrary{fit}
\usetikzlibrary{decorations.pathreplacing}

\allowdisplaybreaks

\newtoggle{quickdraw}
\toggletrue{quickdraw} 

\definecolor{lightgrey}{rgb}{0.9,0.9,0.9}
\definecolor{darkgreen}{rgb}{0,0.3,0}

\definecolorset{rgb}{}{}{darkred,0.8,0,0;darkgreen,0,0.5,0;darkblue,0,0,0.5}

\newtheorem{claim}{Claim}

\renewcommand{\vec}[1]{\mathbf{#1}}

\newcommand{\sequences}{\vec{S}}
\newcommand{\latentData}{\vec{X}}
\newcommand{\latentdata}{\vec{z}}

\newcommand{\phylogeneticParameters}{\boldsymbol{\phi}}
\newcommand{\phylogeny}{{\cal G}}
\newcommand{\tree}{\phylogeny}

\newcommand{\cdensity}[2]{\ensuremath{p(#1 \,|\,#2)}}
\newcommand{\density}[1]{\ensuremath{p(#1 )}}

\newcommand{\treeNode}{\nu}

\newcommand{\nodeIndex}{c}

\newcommand{\rootVarianceScalar}{\tau_0}

\newcommand{\treeVariance}{\vec{V}_{\tree}}

\newcommand{\mdsSD}{\sigma}
\newcommand{\mdsVariance}{\mdsSD^2}

\newcommand{\order}[1]{{\cal O}\hspace{-0.2em}\left( #1 \right)}

\newcommand{\pathLengthNew}[2]{
	d_{F}
	(
	{#1}, {#2}
	)
}

\definecolor{trevorblue}{rgb}{0.330, 0.484, 0.828}
\definecolor{trevoryellow}{rgb}{0.829, 0.680, 0.306}

\title{From viral evolution to spatial contagion: a biologically modulated Hawkes model}
\date{}

\author[1]{Andrew J.~Holbrook}
\author[2]{Xiang Ji}
\author[1,3,4]{Marc A.~Suchard}
\affil[1]{Department of Biostatistics, University of California, Los Angeles}
\affil[2]{Department of Mathematics, Tulane University}
\affil[3]{Department of Biomathematics, University of California, Los Angeles}
\affil[4]{Department of Human Genetics, University of California, Los Angeles}

\begin{document}

\maketitle


\begin{abstract}

Mutations sometimes increase contagiousness for evolving pathogens.  During an epidemic, scientists use viral genome data to infer a shared evolutionary history and connect this history to geographic spread.
We propose a model that directly relates a pathogen's evolution to its spatial contagion dynamics---effectively combining the two epidemiological paradigms of phylogenetic inference and self-exciting process modeling---and apply this \emph{phylogenetic Hawkes process} to a Bayesian analysis of 23,422 viral cases from the 2014-2016 Ebola outbreak in West Africa. The proposed model is able to detect individual viruses with significantly elevated rates of spatiotemporal propagation for a subset of 1,610 samples that provide genome data.  Finally, to facilitate model application in big data settings, we develop massively parallel implementations for the gradient and Hessian of the log-likelihood and apply our high performance computing framework within an adaptively preconditioned Hamiltonian Monte Carlo routine.

\paragraph{Keywords} Bayesian phylogeography; Ebola virus; parallel computing; spatiotemporal Hawkes processes.
\end{abstract}

\section{Introduction}

\newcommand{\citations}{\textcolor{red}{(CITATIONS)} }

 The COVID-19 pandemic has demonstrated the need for new scientific tools for the analysis and prediction of viral contagion across human landscapes.  The mathematical characterization of the complex relationships underlying pathogen genetics and spatial contagion stands as a central challenge of 21st century epidemiology.  We approach this task by unifying two distinct probabilistic approaches to viral modeling.  On the one hand, Bayesian phylogenetics \citep{sinsheimer1996bayesian,yang1997bayesian,mau1999bayesian,suchard01} uses genetic sequences from a limited collection of viral samples to integrate over high-probability shared evolutionary histories in the form of \emph{phylogenies} or family trees.  On the other hand, self-exciting, spatiotemporal Hawkes processes \citep{reinhart2018review} model spatial contagion by allowing an observed event to increase the probability of additional observations nearby and in the immediate future.

Both modeling paradigms come with their own advantages.  For Bayesian phylogenetics, the past twenty years have witnessed a slew of high-impact scientific studies in viral epidemiology \citep{rambaut2008genomic,smith2009origins,faria2014early,gire2014genomic,dudas2017virus,boni2020evolutionary} and the rise of powerful computing tools facilitating inference from expressive, hierarchical models of phylogenies and evolving traits \citep{ronquist2012mrbayes,suchard2018bayesian}.  Unfortunately, the number of evolutionary trees to integrate over explodes with the number of viral samples analyzed \citep{felsenstein1978number}, so Bayesian phylogenetic analyses typically restrict to a relatively small number of viral samples, at most totaling a few thousand.  The fact that viral cases that undergo genetic sequencing usually represent a small subset of the total case count exacerbates this issue.  Thus, failure to detect phylogenetic clades that represent novel strains on account of computational and surveillance limitations always remains a possibility.  Until now, these weaknesses have also held for the sub-discipline of Bayesian phylogeography, which attempts to relate viral evolutionary histories to geographic spread as represented by (typically Brownian) phylogenetic diffusions.  These models describe viral spread through either discretized \citep{lemey2009bayesian} or continuous \citep{lemey2010phylogeography} space, but both approaches induce their own form of bias \citep{holbrook2020massive}.  In the face of these shortcomings, Bayesian phylogeography needs new tools for directly modeling spatial contagion.

\begin{table}[!t]
	\centering
	\begin{tabular}{lll} 
		\toprule
		& Traditional Bayesian phylogenetics &  Hawkes processes  \\
		\midrule
		Observational limit & $N$ in low thousands & $N$ in high tens-of-thousands \\
		Biological insight & Evolutionary history & None \\
		Genetic sequencing & Required & Not required \\
		Spatiotemporal data & Not required & Required \\
		Geographic spread & Not modeled & Modeled\\
		Large-scale transport & Does not induce bias & Induces bias \\
		\bottomrule
	\end{tabular}
	\caption{Comparison of two probabilistic modeling paradigms within viral epidemiology, the combination of which represents a new tool for Bayesian phylogeography.}
\end{table}

Hawkes processes \citep{hawkes1971point,hawkes1971spectra,hawkes1972spectra,hawkes1973cluster,hawkes2018hawkes} are widely applicable point process models for generally viral or contagious phenomena, such as earthquakes and aftershocks \citep{hawkes1973cluster,ogata1988statistical,zhuang2004analyzing,fox2016spatially}, financial stock trading \citep{bacry2015hawkes,hawkes2018hawkes}, viral content on social media \citep{rizoiu2017tutorial,kobayashi2016tideh}, gang violence \citep{mohler2013modeling,mohler2014marked,loeffler2018gun,park2019investigating,holbrook2021scalable} and wildfires \citep{schoenberg2004testing}. Unsurprisingly, Hawkes processes are natural models for the contagion dynamics of biological viruses as well. 
\cite{kim2011spatio} uses spatiotemporal Hawkes processes  \citep{reinhart2018review}, which model viral cases as unmarked events in space and time, to model the spread of avian influenza virus (H5N1). \citet{meyer2014power} incorporate power laws to describe spatial contagion dynamics and model meningococcal disease in Germany from 2001 to 2008. While \citet{rizoiu2018sir} do not model epidemiological data, they do draw connections between epidemiological SIR models and Hawkes processes, showing that the rate of events in the SIR model is equal to that of a finite-population Hawkes model.  \citet{kelly2019real} apply a temporal nonparametric Hawkes process to the 2018-2019 Ebola outbreak in the Democratic Republic of the Congo and successfully generate accurate disease prevalence forecasts. \citet{chiang2020hawkes} model COVID-19 cases and deaths in the U.S. at the county level using spatially indexed mobility and population data to modify the process conditional intensity.  Most recently, \citet{bertozzi2020challenges} compare the performance of a temporal Hawkes process model with temporally evolving conditional intensity to that of SIR and SEIR models for modeling regional COVID-19 case dynamics.

 Because such Hawkes processes do not involve genetic information, one may apply the model to a much larger collection of cases, i.e., those for which a timestamp and spatial coordinates are available.  Moreover, recent successes in scaling Hawkes process inference to a big data setting enable inference from observations numbered in the high tens-of-thousands \citep{holbrook2021scalable,Yuan2021FastEO}.
This ability to interface with an order of magnitude more cases represents a major benefit of the Hawkes process in comparison to the Bayesian phylogenetic paradigm, but the trade-off is that conclusions drawn from a Hawkes process analysis are devoid of explicit biological insight.  It is possible for the model to attribute self-exciting dynamics to nearby viral cases that are only distantly related on the phylogenetic tree. Finally, these processes do not immediately account for viral spread through large-scale transportation networks \citep{brockmann2013hidden,holbrook2020massive} but attribute events resulting from such contagion to a `background' process.

In the following, we construct a Bayesian hierarchical model that allows both modeling approaches to support each other.  This model (Figure \ref{fig:mod}) learns phylogenetic trees that describe the evolutionary history of the subset of observations that yield genetic sequencing and uses this history to inform the distribution of a latent relative rate or \emph{productivity} \citep{schoenberg2019recursive,schoenberg2020nonparametric} for each virus in this limited set. In turn, these virus-specific rates modify the rate of self-excitation of a spatiotemporal Hawkes process describing the contagion of all viruses, sequenced or not.  We use a Metropolis-within-Gibbs strategy to jointly infer all parameters and latent variables of our phylogenetic Hawkes process and overcome the $O(N^2)$ computational complexity of the Hawkes process likelihood by incorporating the modified likelihood in the \textsc{hpHawkes} open-source, high-performance computing library \citep{holbrook2021scalable} available at \url{https://github.com/suchard-group/hawkes}.  Within the same library, we also develop multiple parallel computing algorithms for the log-likelihood gradient and Hessian with respect to the model's virus-specific rates.  GPU-based implementations of these gradient and Hessian calculations score hundred-fold speedups over single-core computing and help overcome quadratic complexity in the context of an adaptively preconditioned Hamiltonian Monte Carlo \citep{neal2011mcmc}.  These speedups prove useful in our analysis of 23,422 viral cases from the 2014-2016 Ebola outbreak in West Africa.

\section{Methods}

\newcommand{\ttheta}{\boldsymbol{\theta}}

\newcommand{\x}{\mathbf{x}}
\newcommand{\dd}{\mbox{d}}
\newcommand{\Id}{\mathbf{I}}
\newcommand{\one}{\boldsymbol{1}}

We develop the phylogenetic Hawkes process and its efficient inference in the following sections.  Importantly, our proposed hierarchical model integrates both sequenced and unsequenced viral case data, representing a significant and clear contribution insofar as:
\begin{enumerate}
	\item the percentage of confirmed viral cases sequenced during an epidemic is often in the single digits \citep{wadman2021united};
	\item and previous phylogeographic models have failed to leverage additional information provided by geolocated, unsequenced case data.
\end{enumerate}
We address this shortcoming by constructing a new hierarchical model that \emph{both} models all spatiotemporal data with a Hawkes process (Section \ref{sec:hawkes}) \emph{and} allows an inferred evolutionary history in the form of a phylogenetic tree to influence dependencies between relative rates of contagion (Section \ref{sec:phyloBrown}) for the small subset of viral cases for which genome data are available.  We believe that this approach is altogether novel.

\subsection{Spatiotemporal Hawkes process for viral contagion}\label{sec:hawkes}
\newcommand{\ttimes}{\mathbf{t}}

Hawkes processes \citep{hawkes1971point,hawkes1971spectra,hawkes1972spectra,hawkes2018hawkes} constitute a useful class of inhomgeneous Poisson point processes \citep{daley2003introduction} for which individual events contribute to an increased rate of future events.  Spatiotemporal Hawkes processes \citep{reinhart2018review} are marked Hawkes processes with spatial coordinates for marks \citep{daley2003introduction}.  We are interested in spatiotemporal Hawkes processes with infinitesimal rate
\begin{align*}
\lambda(\x,t) = \mu(\x) + \xi(\x,t) = \mu(\x) + \sum_{t_n<t} g_n(\x - \x_n, t - t_n)  \, ,
\end{align*}
where $\x \in \mathbb{R}^D$, $t \in \mathbb{R}^+$ and the subscript $n$ indicates that the usual triggering function $g(\cdot,\cdot)$ takes on different forms depending on some characteristic associated with event $n$.  These non-negative, monotonically non-increasing, event-indexed triggering functions additively contribute to $\xi(\cdot,\cdot)$, the \emph{self-excitatory} rate component, and encourage this rate to increase after each observed event.  Here, $\mu(\cdot)$ is the \emph{background} rate and only depends on spatial position $\x$.  Conditioned on observations $(\x_n,t_n)$, $n=1,\dots,N$, we specify the rate components
\begin{align*}
\mu(\x) = \frac{\mu_0}{\tau_x^D} \sum_{n=1}^N \phi\left(\frac{\x-\x_n}{\tau_x}\right)  \mathcal{I}_{[x\neq x_n]} \, \quad \mbox{and} \quad \xi(\x,t) = \frac{\theta_0 \omega}{ h^D} \sum_{t_n<t} \theta_n\, e^{- \omega\, (t-t_n) }  \phi\left(\frac{\x-\x_n}{h}\right) \, ,
\end{align*}
where $\tau_x>0$ and $h>0$  are the background and self-excitatory spatial lengthscales,  $\mu_0>0$ and $\theta_0>0$  are the background and self-excitatory weights, $1/\omega>0$ is the self-excitatory temporal lengthscale, $\phi(\cdot)$ is the $D$-dimensional standard normal probability density function, and the background rate's indicator function prevents a trivial maximum at $\tau_x \rightarrow 0$ \citep{habbema1974stepwise,robert1976choice}.   The inclusion of $\theta_n>0$ for $n=1,\dots,N$ within the self-excitatory rate marks a major departure from similar model specifications in \citet{loeffler2018gun,holbrook2021scalable}.  These `degrees of contagion' or `productivities' \citep{schoenberg2019recursive,schoenberg2020nonparametric} allow different events to contribute differently to the overall self-excitatory rate of the process: the larger the $\theta_n$, the higher the rate directly following event $n$ (Figure \ref{fig:mod}, \emph{right}).  Following the connection of the Hawkes process with exponential triggering function to a discrete time SIR model \citep{rizoiu2018sir}, \citet{bertozzi2020challenges} refer to these quantities as a reproduction number.  In the following, we refer to $\theta_n$ as the \emph{event-specific}, \emph{case-specific} or \emph{virus-specific rate} for the $n$th event, case or viral observation.

 Denoting $\ttheta = (\theta_1,\dots,\theta_N)^T$, the likelihood of observing data $(\latentData,\ttimes)=\left((\x_1,t_1), ..., (\x_N,t_N)\right)^T$ is \citep{daley2003introduction}
\begin{align*}
\mathcal{L}(\latentData,\ttimes|\mu_0,\tau_x,\theta_0,\ttheta,\omega,h)
= \exp \left( - \int_{\mathbb{R}^D} \int_0^{t_N} \lambda(\x,t) \, \dd t\, \dd\x  \right)  \prod_{n=1}^N \lambda(\x_n,t_n) := e ^{ - \Lambda(t_N) } \cdot \prod_{n=1}^N \lambda_n  \, .
\end{align*}
 The choice of $\mathbb{R}^D$ for integration domain is popular and often necessary but assumes complete observation over the entirety of $\mathbb{R}^D$ \citep{schoenberg2013facilitated}.  The resulting integral may be written (Appendix \ref{sec:integral})
\begin{gather*}
\Lambda(t_N) = \mu_0t_N - \theta_0 \sum_{n=1}^N \theta_n \left( e^{-\omega\, (t_N-t_n)} -1 \right) := \sum_{n=1}^N \Lambda_n   \, ,
\end{gather*}
leading to a log-likelihood of
\begin{align}\label{eq:likelihood}
\ell(\latentData,\ttimes|\mu_0,\tau_x,\theta_0,\ttheta,\omega,h) &= - \Lambda(t_N) + \sum_{n=1}^N\log \lambda_n   \\ \nonumber
&= \sum_{n=1}^N\Bigg\{ \log \Bigg[  \sum_{n'=1}^N \Bigg( \frac{\mu_0\, \mathcal{I}_{[\x_n\neq \x_{n'}]}}{\tau_x^D} \phi\left(\frac{\x_n-\x_{n'}}{\tau_x}\right)  \\ \nonumber
& \hspace{4em} +\frac{\theta_0\theta_{n'} \omega\, \mathcal{I}_{[t_{n'}<t_n]}}{ h^D}e^{- \omega\, (t_n-t_{n'})}   \phi\left(\frac{\x_n-\x_{n'}}{h}\right) \Bigg)\Bigg]   - \Lambda_n \Bigg\} \\ \nonumber
&:= \sum_{n=1}^N \left[
\log \left(  \sum_{n'=1}^N \lambda_{nn'} \right)  - \Lambda_n \right] := \sum_{n=1}^N \ell_n  \, .
\end{align}
We reference these formulas while outlining our inference strategy in Section \ref{sec:inference} and detailing our massively parallel algorithms for calculating the log-likelihood gradient and Hessian with respect to event-specific rates $\ttheta$ in Appendix \ref{sec:parallel}. We describe our biologically modulated joint prior on event-specific rates $\theta_1, \dots, \theta_N$ in the next section.

\subsection{Phylogenetic Brownian process prior on rates}\label{sec:phyloBrown}

\begin{figure}
	\centering
	\includegraphics[width=\linewidth]{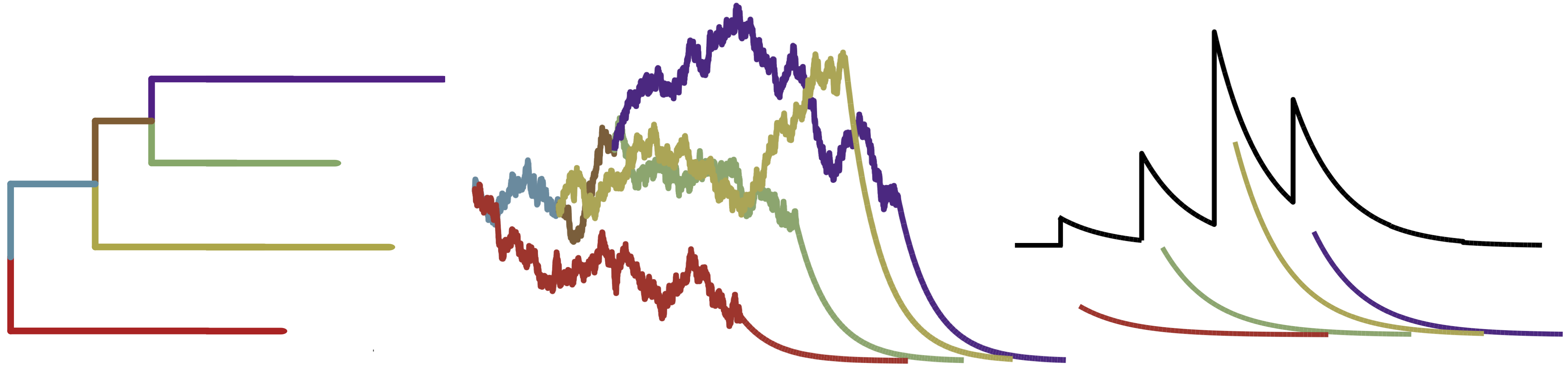}
	\caption{The phylogenetic Hawkes process relates the evolutionary history of a virus to the rate at which subsequent viral cases occur nearby. Left: a phylogenetic tree characterizes the evolution of a viral strain. Middle: a Brownian motion `along the tree' describes the evolution of infinitesimal rates as a function of branch lengths and tree topology (Section \ref{sec:phyloBrown}). Right: virus-specific rates additively contribute to the rate of occurrence for future cases (Section \ref{sec:hawkes}).}\label{fig:mod}
\end{figure}

\newcommand{\whichPhylogeny}[1]{\phylogeny\left( #1 \right)}
\newcommand{\bbeta}{\boldsymbol{\beta}}
\newcommand{\M}{\mathcal{M}}

We use standard Bayesian phylogenetics hierarchical approaches \citep{Suchard03HPM} to model a single molecular sequence alignment $\sequences$ containing sequences from $M\leq N$ evolutionarily related viruses. Let $\M$ denote the ordered index set with cardinality $|\M|=M$ containing every number within the set $\{1,\dots,N\}$ that corresponds to an observed virus for which genome data are present.  In the following, we number the elements within $\M$ as $m_1, m_2, \dots, m_M$. Moreover, we make use of the set $\M^+$ with cardinality $|\M^+|= 2M-1$, satisfying $\M \subset \M^+$ and containing elements $m_1, \dots, m_{2M-1}$.
Our primary object of interest is the phylogenetic tree $\phylogeny$ (Figure \ref{fig:mod}, \emph{left})  defined as a bifurcating, directed graph with $M$ terminal degree-1 nodes $(\treeNode_{m_1}, \ldots, \treeNode_{m_M})$ that correspond to the tips of the tree (or sequenced observations), $M - 2$ internal degree-3 nodes $(\treeNode_{m_{M+1}}, \ldots, \treeNode_{m_{2M-2}})$, a root degree-2 node $\treeNode_{m_{2M-1}}$ and edge weights $(w_{m_1},\ldots, w_{m_{2M-2}})$ that encode the elapsed evolutionary time between nodes. Here, each $w_m$ communicates the expected number of molecular substitutions per site, which is itself the product between the real time duration and the evolutionary rate arising from a molecular clock model. For example, we use a relaxed molecular clock model \citep{drummond2006relaxed} that allows for substitution rates to flexibly vary accross branches (Section \ref{sec:ebola}).  One may either know $\phylogeny$ \emph{a priori} or endow it with a prior distribution parameterized by some vector $\phylogeneticParameters$.
\citet{suchard01} and \cite{suchard2018bayesian} develop the joint distribution $\density{\sequences, \phylogeneticParameters, \phylogeny}$ in detail.

We assume that the event-specific rates $\ttheta$ defined within our Hawkes model take the form (Figure \ref{fig:mod}, \emph{middle})
\begin{align*}
 \begin{cases}
\theta_n = \theta_n(z_n) = \exp\left(z_n+\boldsymbol{\beta}^T \mathbf{f}(t_n)\right) & z_n \in \mathbb{R}\, , \quad n\in \M \\
\theta_n = 1 & \quad n\notin \M  \, ,\\
\end{cases}
\end{align*}
and that the elements of vector $\latentdata = (z_{m_1}, \dots, z_{m_M})^T$ follow a Brownian diffusion process along the branches of $\phylogeny$ \citep{cavalli1967phylogenetic,felsenstein85,lemey2010phylogeography}.
Here, $\mathbf{f}(\cdot)$ is some fixed vector function and the inclusion of the linear term $\boldsymbol{\beta}^T \mathbf{f}(t_n)$ helps control for global trends resulting from extrinsic events such as mass quarantine or travel restrictions.
Under the Brownian process, the latent value of a child node $\treeNode_{\nodeIndex}$ in tree $\phylogeny$ is normally distributed about the value of its parent node $\treeNode_{\text{\tiny pa}(\nodeIndex)}$ with variance $w_{\nodeIndex} \times \sigma^2$, where $\sigma^2$ gives the dispersal rate after controlling for correlation in values that are shared by descent through the phylogenetic tree $\phylogeny$.
We further posit that the latent value of the root node $\treeNode_{m_{2M-1}}$ is \emph{a priori} normally distributed with mean $0$ and variance $\rootVarianceScalar \times \sigma^2$.
The vector $\latentdata$ is then multivariate normally distributed \citet{cybis2015assessing} and has probability density function
\newcommand{\sameTree}[2]{\delta_{#1 #2}}
\begin{align} \small
\cdensity{\latentdata}{ \treeVariance,  \sigma^2, \rootVarianceScalar}
=	\left(2 \pi\sigma^2 \right) ^{-M/2}
\left|
\treeVariance
\right|^{-1/2}
	\mbox{exp}
	\left(
	-\frac{1}{2\sigma^2}
\latentdata^{T}
	\treeVariance^{-1}
 \latentdata
	\right)
\,  ,
\label{eq:multinormal}
\end{align}
where
$\treeVariance = \{ v_{nm} \}$ is a symmetric, positive definite, block-diagonal $M \times M$ matrix with structure dictated by $\phylogeny$. Defining $\pathLengthNew{u}{v}$ to be the sum of edge-weights along the shortest path between nodes $u$ and $v$ in tree $\tree$, the diagonal elements $v_{mm} = \rootVarianceScalar + \pathLengthNew{\treeNode_{m_{2M-1}}}{\nu_{m_m}}$ are the elapsed evolutionary time between the root node and tip node $m_m$, and off-diagonal elements $v_{nm} = \rootVarianceScalar +
\left[
\pathLengthNew{\treeNode_{m_{2M-1}}}{\nu_{m_n}} + \pathLengthNew{\treeNode_{m_{2M-1}}}{\nu_{m_m}}
- \pathLengthNew{\nu_{m_n}}{\nu_{m_m}}
\right] / 2$ are the evolutionary time period between the root node
and the most recent common ancestor of tip nodes $m_n$ and $m_m$.

\subsection{Inference}\label{sec:inference}

\newcommand{\momentum}{\mathbf{p}}
\newcommand{\mass}{\mathbf{M}}

Due to the complexity of the phylogenetic Hawkes process and the large number of viruses we seek to model, we must use advanced statistical, algorithmic and computational tools to infer the posterior distribution
\begin{align}
\label{eq:posterior}
\cdensity{ \mdsVariance, \tree, \phylogeneticParameters,\mu_0,\tau_x,\theta_0,\ttheta,\omega,h,\boldsymbol{\beta} }{\latentData, \ttimes, \sequences}
 \propto &
\mathcal{L}(\latentData,\ttimes|\mu_0,\tau_x,\theta_0,\ttheta(\latentdata),\omega,h) \times
\cdensity{\latentdata}{\mdsVariance, \tree} \times \\  \nonumber
&
\density{\mu_0} \times
\density{\tau_x} \times
\density{\theta_0} \times
\density{\omega} \times
\density{h} \times \density{\boldsymbol{\beta}} \times \\ \nonumber
&\density{\mdsVariance}
\times \density{\sequences, \phylogeneticParameters, \tree}  \, .
\end{align}
We do so using a random-scan Metropolis-with-Gibbs scheme, in which we compute key quantities with the help of adaptively preconditioned Hamiltonian Monte Carlo (HMC) \citep{neal2011mcmc},
dynamic programming and parallel computing on cutting-edge graphics processing units (GPU).

\subsubsection{Dynamic programming for phylogenetic diffusion quantities}

We must evaluate $\cdensity{\latentdata}{\mdsVariance, \tree}$ to sample $\tree$.
The bottleneck within the evaluation of Equation \eqref{eq:multinormal} is the ostensibly $\order{ M^3 }$ matrix inverse $\treeVariance^{-1}$, but \citet{pybus2012unifying} develop a dynamic programming algorithm to perform the requisite computations in $\order{M}$ with parallelized post-order traversals of $\tree$.  We use this algorithm, which is closely related to the linear-time algorithms of \citet{freckleton2012fast} and \citet{ho2014linear}, as all are examples of message passing on a directed, acyclic graph \citep{cavalli1967phylogenetic,pearl1982reverend}.
Similar tricks render inference for $\phylogeneticParameters$ linear in $M$, and \citet{fisher2021relaxed} extend \citet{pybus2012unifying} to compute gradients with respect to $\phylogeneticParameters$.
Finally, implementing these algorithms on GPUs would lead to additional speedups \citep{suchard2009many}, but the computational bottleneck we face when applying the phylogenetic Hawkes process arises from the Hawkes process likelihood and its gradients.

\subsubsection{Massive parallelization for Hawkes model quantities}

Sampling the Hawkes process parameters $\mu_0, \tau_x, \theta_0,\omega, h, \beta$ and event-specific rates $\ttheta$ requires evaluation of the likelihood $\mathcal{L}(\latentData,\ttimes|\mu_0,\tau_x,\theta_0,\ttheta,\omega,h)$ or its logarithm.  Unfortunately, the double summation of Equation \eqref{eq:likelihood} results in an $\order{N^2}$ computational complexity that makes repeated likelihood evaluations all but impossible for the number of observations considered in this paper.  We therefore use the high-performance computing framework of \citet{holbrook2021scalable} to massively parallelize likelihood evaluations in the context of univariate, adaptive Metropolis-Hastings proposals for parameters $\mu_0, \tau_x, \theta_0,\omega$, $h$ and $\beta$.
On the other hand, inference for the $M$-vector $\latentdata$ requires more than fast univariate proposals, so we opt for HMC to sample from its high-dimensional posterior.  Even in high dimensions, HMC efficiently generates proposal states by simulating a physical Hamiltonian system that renders the target posterior distribution invariant.  Here, we follow standard procedure and specify the system with total energy
\begin{align*}
H(\latentdata,\momentum)= - \log \big(\pi(\latentdata)\, \xi(\momentum|\mass) \big) \propto -\log \pi(\latentdata) + \frac{1}{2} \momentum^T\mass^{-1}\momentum \, ,
\end{align*}
where $\pi(\latentdata)$ is the density of the marginal posterior for $\latentdata$, $\momentum$ is a Gaussian distributed `momentum' variable with density  $\xi(\momentum|\mass)$, and $\mass$ is the system mass matrix and the covariance of $\momentum$.  We again follow standard HMC procedure and use the leapfrog algorithm \citep{leimkuhler2004simulating} to integrate Hamilton's equations
\begin{align*}
\dot{\latentdata} = \frac{\partial}{\partial \momentum}H(\latentdata,\momentum) = \frac{1}{2} \mass^{-1} \momentum\, , \quad 
\dot{\momentum} = - \frac{\partial}{\partial \latentdata}H(\latentdata,\momentum) = \nabla \log \pi(\latentdata) \, .
\end{align*}
But there is no free lunch: simulation of the physical system requires repeated evaluations of the log-likelihood gradient, and these evaluations may become burdensome in big data contexts.  Indeed, the gradient of the Hawkes process log-likelihood of Equation \eqref{eq:likelihood} with respect to $\ttheta$ becomes computationally onerous for large $N$ and $M$. The gradient with respect to a single event-specific rate $\theta_m$ takes the form
\begin{align}\label{eq:gradient}
\frac{\partial \ell}{\partial \theta_m} &= - \frac{\partial \Lambda_m}{\partial \theta_m} + \sum_{t_m<t_{n}} \frac{1}{\lambda_{n}}  \frac{\partial \lambda_{nm}}{\partial \theta_m} \\ \nonumber
&= \theta_0  \left(e^{- \omega\, (t_{N}-t_{m}) } - 1 \right) +  \sum_{t_m<t_{n}} \frac{1}{\lambda_{n}}  \frac{\theta_0 \omega}{h^D} e^{- \omega\, (t_{n}-t_m) }  \phi\left(\frac{\x_{n}-\x_m}{h}\right) \, ,
\end{align}
where the summation and $\lambda_{n}$ are both of complexity $\order{N}$.  Thus, computing the entire vector $\partial \ell /\partial \ttheta = (\partial \ell/\partial \theta_1,\dots ,  \partial \ell/\partial \theta_M)^T$ requires time $\order{NM}$.  Worse still, due to the multiscale nature of the posterior for the relative rates (Figure \ref{fig:diagnostic}), we find it necessary to precondition the Hamiltonian dynamics by specifying a diagonal mass matrix with elements
\begin{align}\label{eq:hess}
\mass_{mm}^{-1} \approx - \frac{\partial^2 \ell}{\partial \theta^2_m} = \sum_{t_m<t_{n}} \frac{1}{\lambda_{n}^2}  \frac{\theta^2_0 \omega^2}{h^{2D}} e^{-2 \omega\, (t_{n}-t_m) }  \phi^2\left(\frac{\x_{n}-\x_m}{h}\right) \, .
\end{align}
Specifically, we maintain a running average of Hessians calculated at a fixed interval and use this as our preconditioner $\mass$, thus maintaining asymptotic unbiasedness of Monte Carlo estimates \citep{haario2001adaptive}.
Just as with the gradient, the summation and $\lambda_{n}$ are both of complexity $\order{N}$, and the resulting complexity for the entire Hessian is $\order{NM}$.
To overcome these rate-limiting steps, we develop massively parallel central processing unit (CPU) and GPU implementations of both the gradient and the Hessian.  Although our GPU-based implementations are fastest (Section \ref{sec:perform}), our CPU implementations are competitive, making use of both multi-core processing and SIMD (single instruction, multiple data) vectorization \citep{holbrook2020massive}. Regardless of implementation, all of our high-performance software remains freely available for public use.

%

\subsection{Software availability}

We use the Bayesian evolutionary analysis by sampling trees (\textsc{BEAST}) software package \citep{suchard2018bayesian}, a popular tool for viral phylogenetic inference that implements MCMC methods to explore $\density{\sequences, \phylogeneticParameters, \tree}$ and $\cdensity{\latentdata}{\mdsVariance, \tree}$ \citep{cybis2015assessing} under a range of evolutionary models.
In writing this paper, we have contributed to the open-source, stand-alone library \textsc{hpHawkes} \url{http://github.com/suchard-group/hawkes} for computing the spatiotemporal Hawkes process log-likelihood (Equation \ref{eq:likelihood}), its gradient (Equation \ref{eq:gradient}) and its Hessian (Equation \ref{eq:hess}). \textsc{hpHawkes} integrates into \textsc{BEAST} with the help of an application programming interface (API).
Within \textsc{hpHawkes}, we combine \textsc{C++} code with which standard compilers generate vectorized CPU-specific instructions and \textsc{OpenCL} kernels that allow for GPU-specific optimization.  Finally, we have used the \textsc{Rcpp} package \citep{eddelbuettel2011rcpp} to make the same massive parallelization speedups available to users of the \textsc{R} programming language.

\begin{figure}[t!]
	\centering
	\includegraphics[width=0.9\linewidth]{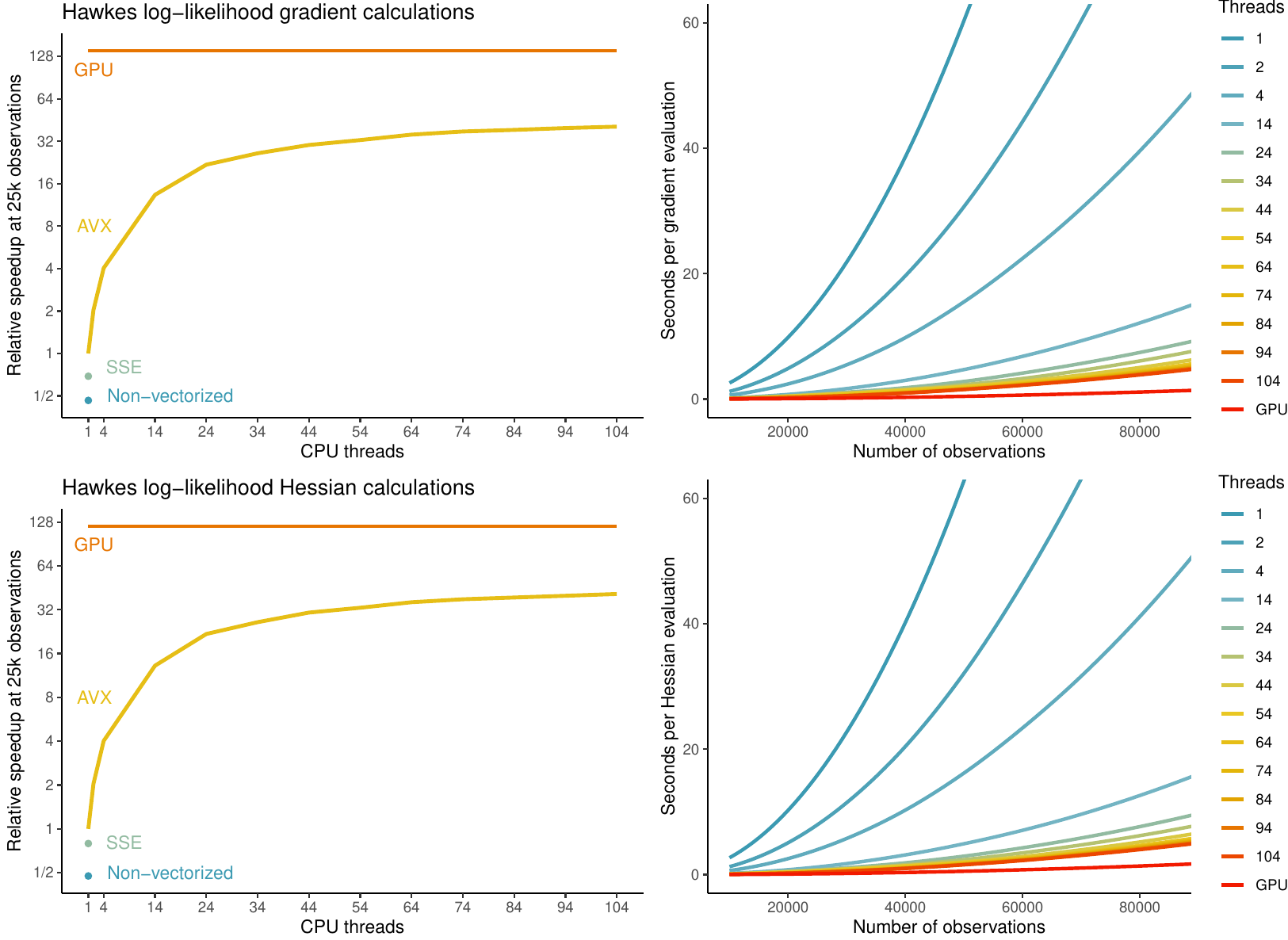}
	\caption{Spatiotemporal Hawkes process log-likelihood gradient and Hessian calculations with respect to event-specific rates $\ttheta$ with central and graphics processing units (CPU and GPU). [Left] Multiplicative speedups over single-threaded advanced vector extensions (AVX) vectorization for single-threaded non-vectorized and streaming SIMD extensions (SSE), multi-threaded AVX and many-core GPU processing for 25,000 randomly generated data points. [Right] Seconds per gradient and Hessian calculations for multi-threaded AVX and GPU implementations from 10 to 90 thousand data points.}\label{fig:speedups}
\end{figure}

\begin{figure}[t]
	\centering
	\includegraphics[width=0.9\linewidth]{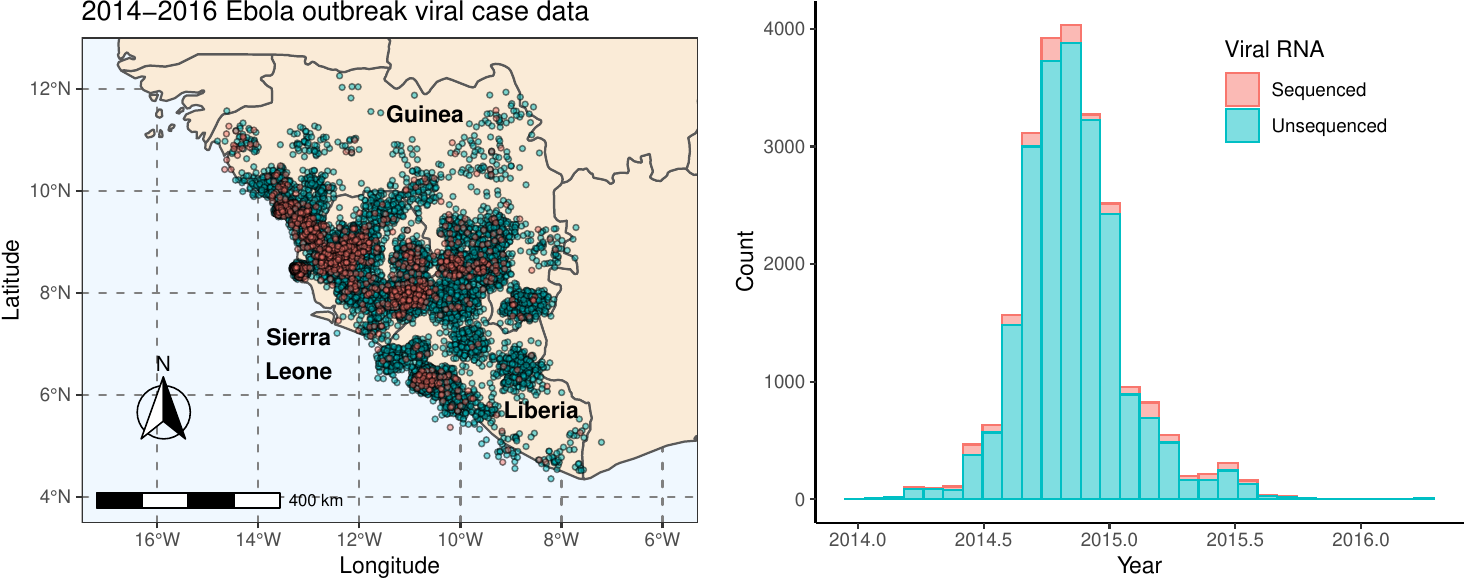}
	\caption{Spatiotemporal distribution of 23,178 viral cases during the 2014-2016 Ebola outbreak in Guinea, Sierra Leone and Liberia. Of this number, 1,367 viral samples yield RNA sequence data and interface directly with the prior over phylogenetic trees. All cases, both sequenced and unsequenced, interface with the Hawkes process likelihood.}\label{fig:data}
\end{figure}

\begin{figure}[t!]
	\centering
	\includegraphics[width=\linewidth]{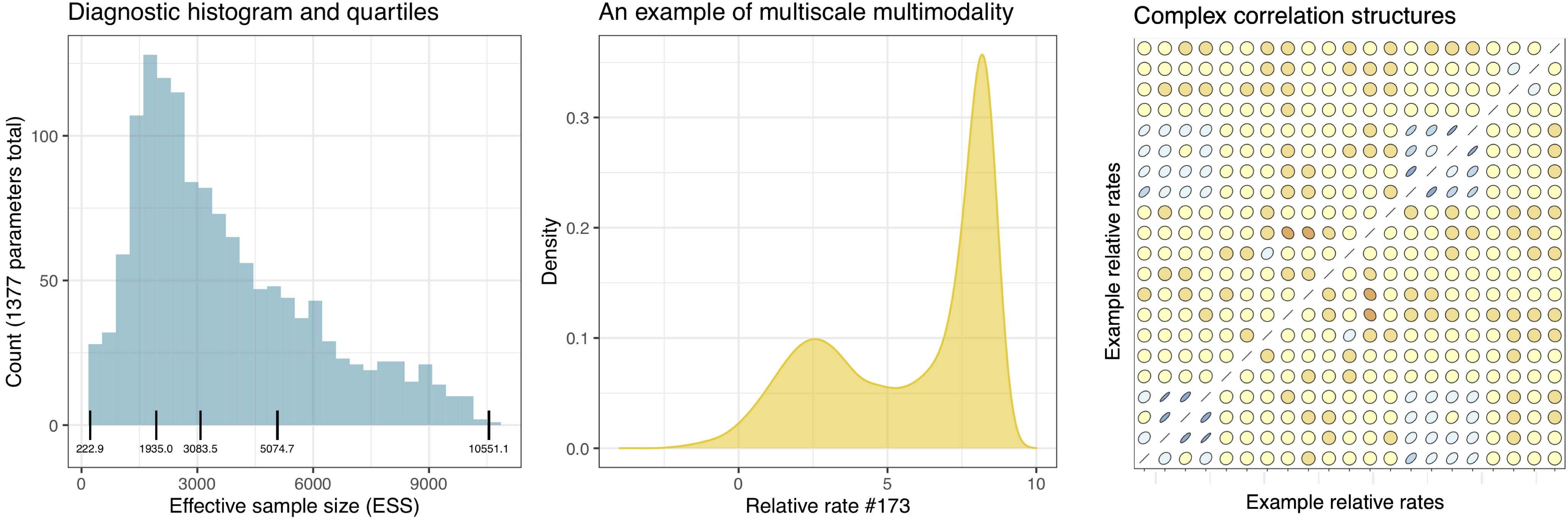}
	\caption{The posterior distribution presents multiple challenges: it is high-dimensional; it takes on different scales for different parameters; it is multimodal in some parameters; and it exhibits complex correlation structures between parameters.  [Left] Histogram and quartiles from 100 million MCMC samples for the ESS of all 1377 model parameters. [Middle] Multimodal marginal posterior for a single relative rate. [Right] Posterior correlations between 21 relative rates.}\label{fig:diagnostic}
\end{figure}

\section{Demonstration}

\subsection{Massive parallelization}\label{sec:perform}

Figure \ref{fig:speedups} shows benchmarking results for evaluating the Hawkes log-likelihood gradient with respect to event-specific rates $\ttheta$ (Equation \ref{eq:gradient}).
 For the GPU results, we use an NVIDIA Quadro GV100, which has 5,120 CUDA cores (at 1.13 GHz) and reaches an (unboosted) 2.9 teraflops peak double-precision floating point performance (or 5.8 teraflops for fused operations such as fused multiply-add).
We use a Linux machine with two 26-core Intel Xeon Gold processors (2.1 GHz) for CPU results.  Each physical core supports 2 threads or logical cores, and the machine achieves a peak performance of 874 gigaflops with double-precision floating point enhanced with AVX vectorization (again, double this for fused operations).
Based on peak double-precision floating point operations, our \emph{a priori} expectation is for fully parallelized GPU-based gradient evaluations to to be roughly 3.3 times faster than 104-threaded AVX evaluations on the CPU.

On the left of Figure \ref{fig:speedups}, we compare relative efficiency for GPU and various CPU implementations of the log-likelihood gradient and Hessian for 25,000 simulated data points using single-threaded AVX computing (15.7 and 16.3 seconds per gradient and Hessian evaluations) as baseline.  Using SSE or non-vectorized single-threaded computing results in 1.5- and 2.2-fold slowdowns for the gradient and 1.3- and 2.1-fold slowdowns for the Hessian. Sticking with AVX processing, we see diminishing returns as we increase the number of threads. For the both the gradient and the Hessian, the 14- 54- and 104-thread AVX implementations are roughly 13, 33 and 41 times faster than single-threaded AVX.  Agreeing with our \emph{a priori} expectations, the GPU implementation is 140.4 times faster than single-threaded AVX and 3.5 times faster than 104-threaded AVX for the gradient and 120.0 times faster than single-threaded AVX and 2.9 times faster than 104-threaded AVX for the Hessian.  The right of Figure \ref{fig:speedups} demonstrates the $\order{N^2}$ computational complexity for the same gradient and Hessian evaluations by varying the number of data points from 10,000 to 90,000.  While parallelization does not overcome this quadratic scaling, it does reduce computational costs for finite observation counts.

\subsection{2014-2016 Ebola outbreak in West Africa}\label{sec:ebola}

\begin{table}[t!]
	\centering
	\resizebox{\textwidth}{!}{\begin{tabular}{lllll} 
			\toprule
			Hierarchical model	&&& Posterior mean & \\
			module & Model parameter & Symbol &(95\% HPD cred. int.) & Unit \\
			\midrule
			Hawkes process  & Background spatial lengthscale &$\tau_x$ &  194 (147, 243) & km\\
			& Self-excitatory temporal lengthscale& $1/\omega$ & 29.8 (28.5, 30.9) & days \\
			& Self-excitatory spatial lengthscale& $h$  &  7.37 (7.13, 7.62)  & km \\
			& Normalized self-excitatory weight& $\theta_0/(\theta_0+\mu_0)$ &  0.96 (0.95, 0.97) & ---\\
			& Temporal trend coefficient &$\beta$            & -2.22 (-2.37,-2.06)  & --- \\
			Phylogenetic diffusion & Standard deviation& $\sigma$ & 51.0 (46.4, 55.7) & log rate  \\
			\bottomrule
	\end{tabular}}
	\caption{Posterior means and 95\% highest posterior density (HPD) credible intervals from the application of the phylogenetic Hawkes process to the 2014-2016 Ebola outbreak in Guinea, Sierra Leone and Liberia.}\label{tab:postRes}
\end{table}

During the 2014-2016 outbreak in Guinea, Sierra Leone and Liberia, Ebola viral fever resulted in over 28,000 known cases and 11,000 known deaths \citep{world2015ebola}.  First reports of the virus in Guinea emerged during March of 2014 \citep{baize2014emergence}.  At around the same time, viral cases with the same Guinean origin \citep{gire2014genomic} emerged in Sierra Leone and Liberia.   In May 2014, the virus crossed from Guinea to Kailahun, Sierra Leone.  From there, it spread to multiple counties of Liberia and Guinea \citep{dudas2017virus}, and the same strain reached Freetown, the capital of Sierra Leone, by July 2014.  In the fall of 2014, Sierra Leone and Liberia were detecting 500 and 700 new cases a week.  Only by the end of 2014 did case numbers begin to abate in most areas due to control measures.  By March of 2015 sustained transmission of the virus only continued in western Guinea and western Sierra Leone \citep{dudas2017virus}.  Figure \ref{fig:data} shows the spatiotemporal distribution of the majority of known Ebola virus cases during the epidemic.

Using our high performance computing framework, we apply the phylogenetic Hawkes process to the analysis of 23,422 viral cases.  \citet{dudas2017virus} provide a total of 1,610 cases furnishing genomic sequencing, 1,367 of which come with date and location data (\url{https://github.com/ebov/space-time}).  We supplement this sequenced data with 21,811 date and location pairs from unsequenced cases documented by the World Health Organization (\url{https://apps.who.int/gho/data/node.ebola-sitrep}).
The precision of the spatial data is district or county level.  To leverage spatial information as much as practically possible within our Hawkes model, we assume the locations follow a Gaussian distribution at district population centroids and with variance guaranteeing a 95\% probability of the case occurring within the circle of equal area to the district and centered at the population centroid.  We then integrate over uncertainty with respect to these locations by periodically sampling new locations according to the assumed Gaussian distribution throughout the MCMC run and with a period of roughly 100 iterations. That said, sensitivity analyses show that model inference is robust to fixing randomly generated locations for the entire MCMC chain.
We make the combined data and documentation for our entire \textsc{BEAST} analysis available within the single file \texttt{Final.xml} and place this as well as other project scripts together at the repository \url{https://github.com/suchard-group/EBOVPhyloHawkes}.  In addition to the software mentioned in the previous section, we make use of the \textsc{ggplot2} and \textsc{ggmap} \textsc{R} packages for data and results visualization \citep{ggplot,kahle2013ggmap}.

\begin{figure}[t]
	\centering
	\includegraphics[width=0.9\linewidth]{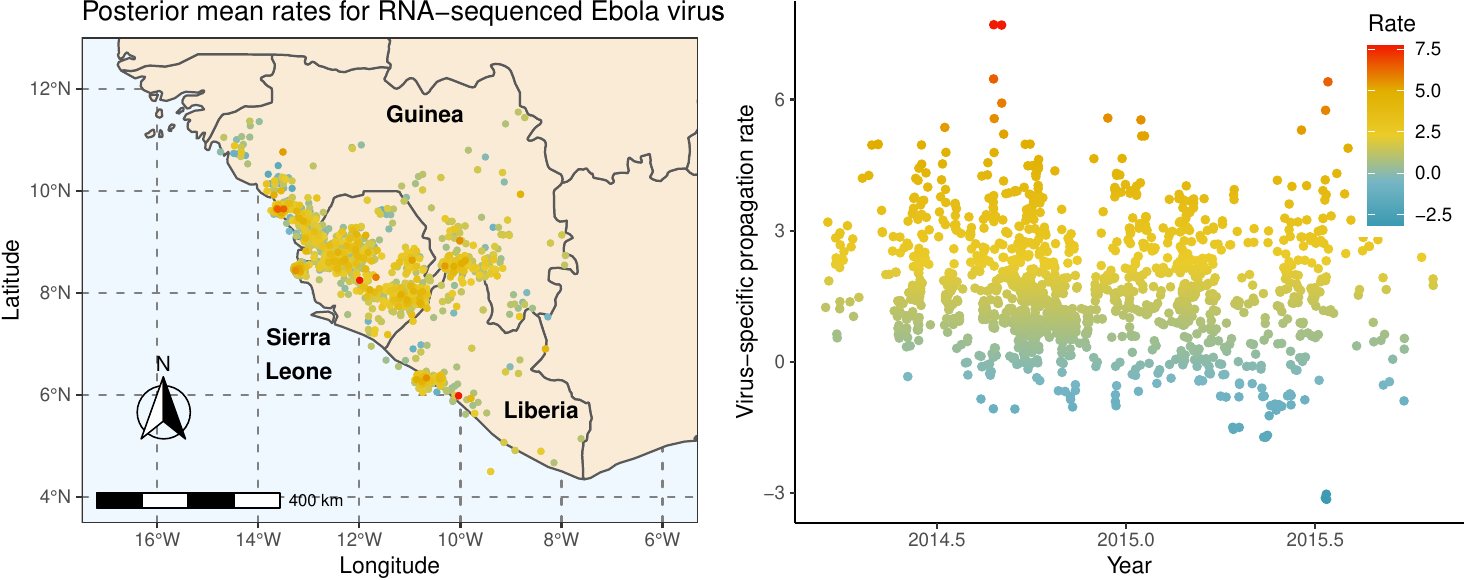}
	\caption{Hawkes model posterior mean rates $\ttheta$ for the 1,367 (of 1,610) RNA-sequenced viral samples for which date/location data are available.  Unsurprisingly, the largest relative rates occur within or nearby major clusters of events. Adjusting for downward trends in case data with a negative coefficient $\beta$ (Table \ref{tab:postRes}) allows detection of higher relative rates after peak outbreak (late 2014) including a jump in infections mid-2015 (Figure \ref{fig:data}).}\label{fig:ratesMap}
\end{figure}

For the phylogenetic prior specification $\density{\sequences, \phylogeneticParameters, \tree}$, we follow the phylogeographic analysis of \citet{dudas2017virus} and use a mixture of 1,000 phylogenetic trees obtained as high-probability posterior samples from their purely phylogenetic analysis of the 1,610 sequenced viral samples.  In that preceding Bayesian analysis, \citet{dudas2017virus} combine an HKY$+\Gamma_4$ substitution model prior for molecular evolution \citep{hasegawa1985dating,yang1994maximum}, a relaxed molecular clock prior on rates \citep{drummond06}, a non-parametric coalescent `Skygrid' prior on effective population size dynamics \citep{gill2013improving} and a continuous time Markov chain reference prior for overall rate \citep{ferreira2008bayesian}.
We assume \emph{a priori} that the background lengthscale $\tau_x$ follows a diffuse inverse gamma distribution with shape 1 and scale 10, where distance units are latitudinal and longitudinal degrees.  An inverse gamma distribution with shape and scale parameters equal to 2 and 0.5 for both $h$ and $1/\omega$ encodes our beliefs that self-excitatory dynamics occur at finer spatiotemporal scales, where years are the temporal units.  We upweight self-excitatory dynamics by giving $\theta_0$ and $\mu_0$ gamma priors with shape parameters $1$ and $2$ and scale parameters $0.001$ and $2$, respectively.  We absorb $\tau_0$ into $\sigma$ and place a tight inverse gamma prior on $1/\sigma$ with shape and scale parameters of 2 and 0.5.  Finally, we set $\mathbf{f}(t_n)=t_n$ and place a normal prior on the univariate coefficient $\beta$ with mean 0 and standard deviation of 10.  We find all parameters robust to prior specification due to the large number of observations considered.

\begin{figure}[t]
	\centering
	\includegraphics[width=\linewidth]{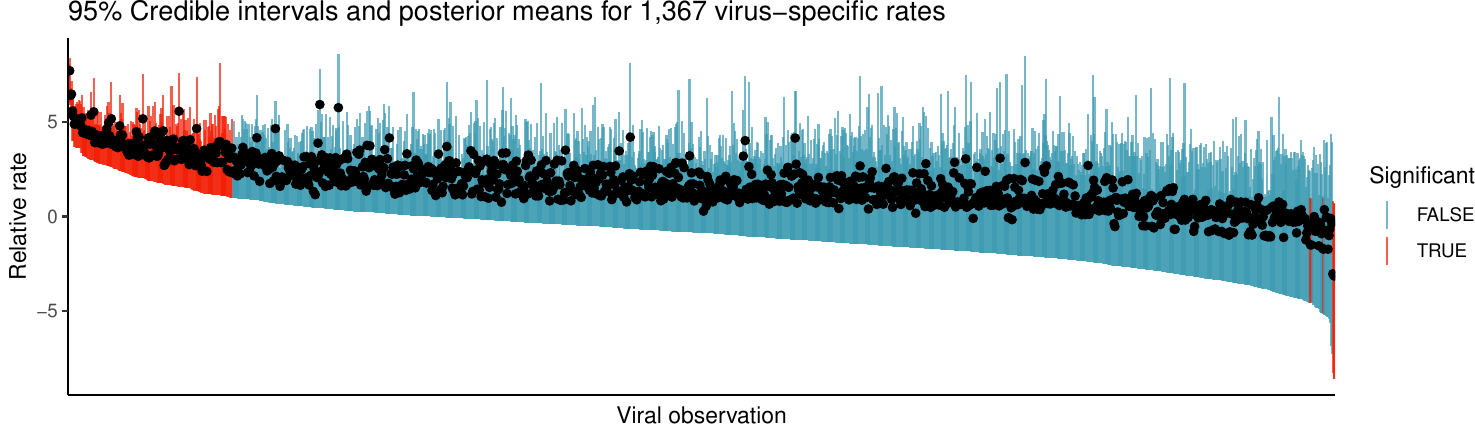}
	\caption{95\% credible intervals and posterior means for virus-specific rates $\ttheta$ corresponding to the subset of 1,610 sequenced viruses that come with date/location data and therefore appear in the Hawkes process module.  We call those 183 intervals which do not include 1 `significant'.  Of these, 177 are above and 6 below 1.  Appendix \ref{sec:signifTaxa} has labels, locations and times for each.}\label{fig:cis}
\end{figure}

\begin{figure}[t]
	\centering
	\includegraphics[width=0.7\linewidth]{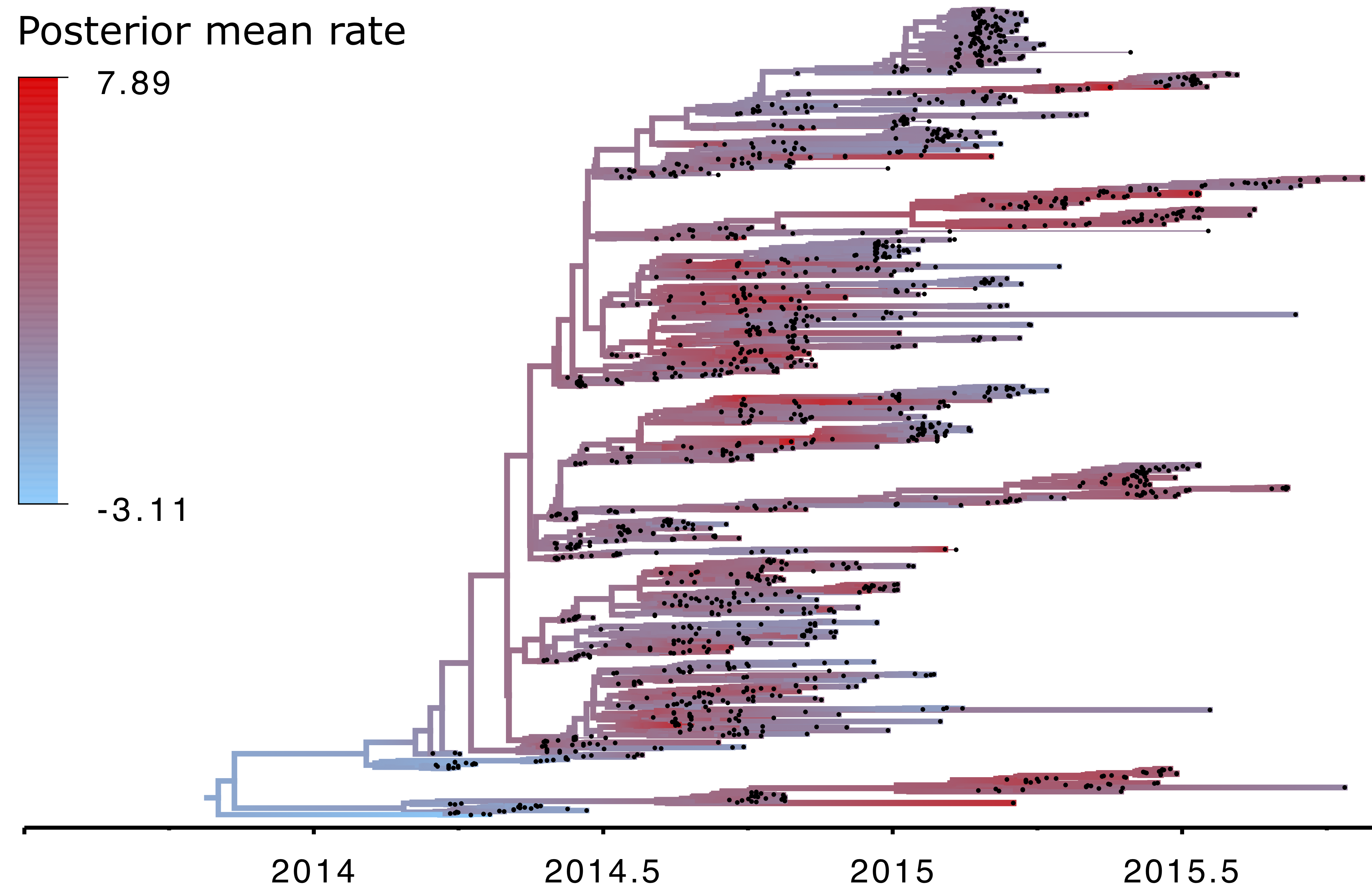}
	\caption{Posterior mean virus-specific relative rates color the posterior maximum clade credibility tree of phylogeny $\tree$, a few subtrees of which potentially demonstrate elevated contagiousness.}\label{fig:tree}
\end{figure}

We generate 100 million MCMC samples according to the routine outlined in Section \ref{sec:inference} and discard the first 500 thousand as burn-in. Using our parallel computing algorithms and a single NVIDIA GV100 GPU (Section \ref{sec:perform}), the routine requires 6.77 hours to generate 1 million samples and 28 days to generate all 100 million samples.  Figure \ref{fig:diagnostic} shows the distribution of effective samples sizes (ESS) across all model parameters and illustrates some of the challenges facing any MCMC routine for the phylognetic Hawkes model.  Namely, the posterior distribution is high-dimensional, multimodal, multiscale and has complex correlation structures.

Table \ref{tab:postRes} shows posterior means and 95\% highest posterior density (HPD) credible intervals for the phylogenetic Hawkes process parameters.  The posterior mean for the spatial bandwidth $\tau_x$ of the Hawkes background process is 194 km (147, 243), allowing the model to incorporate and adapt to large scale geographic movement.  On the other hand, the Hawkes process self-excitatory spatial bandwidth $h$ has a posterior mean of 7.4 km (7.1, 7.6), indicating the smaller local scale for which the model attributes viral contagion.  The self-excitatory temporal bandwidth $1/\omega$ has a posterior mean of 29.8 days (28.5, 30.9), indicating the timescale for which the model attributes the same viral contagion.  The normalized self-excitatory weight $\theta_0/(\theta_0+\mu_0)$ indicates the proportion of events the model attributes to self-excitatory (compared to background) dynamics and has a posterior mean of 0.96 (0.95, 0.97).  The posterior mean of the self-excitatory rate's temporal trend coefficient $\beta$ is -2.22 (-2.37, -2.06) indicates that, for every additional year and \emph{ceteris paribus}, one should expect a multiplicative decrease of $1-\exp(-2.22)\times 100 \approx 90\%$ to the process self-excitatory rate.  In this way, the model adjusts for downward trends arising from epidemiological control (e.g., mass quarantine and travel restrictions) and controls for these factors when inferring virus-specific relative rates.

Next, we consider posterior inference of the virus-specific rates $\ttheta$ for those viral observations that provide RNA sequences.  When interpreting these results, it is important to understand that the phylogenetic Hawkes process implicitly assumes that such samples spark nearby contagion as described by spatial and temporal bandwidths $h$ and $1/\omega$.  Recall that the posteriors for these two parameters concentrate at over 7 km and 4 weeks, respectively.  Figure \ref{fig:ratesMap} depicts the relationship between posterior mean values of $\ttheta$ and the spatiotemporal distribution of corresponding viruses.
Since these rates represent multiplicative factors of the global self-excitatory weight $\theta_0$, a null value would be 1.  Posterior means range from approximately -2.5 to 7.5 and increasingly vary as a function of time.  As one might expect, the highest rates appear near or within larger clusters.  Thanks to the negative temporal trend coefficient $\beta$ and the increase of uncertainty with time, larger rate values do obtain for some viral cases occurring in 2015, despite following after peak epidemic.  Figure \ref{fig:cis} features posterior means and 95\% intervals for the same virus-specific rates.  Only a small subset of 183 rate intervals do not include 1. Of these, 177 have lower bound greater than 1, and 6 have upper bound below 1.  We interpret all 103 of the corresponding viruses as having statistically significantly increased or decreased contagiousness.  Virus-specific labels, locations and times for all 183 appear in Appendix \ref{sec:signifTaxa}.

Finally, Figure \ref{fig:tree} shows how these posterior rates organize as a function of the inferred posterior maximum clade credibility tree $\tree$.  Generally speaking, shorter branch lengths indicate larger effective populations of viruses, while larger branch lengths indicate smaller.  For example, the structure of the bottom subtree reflects this intuition as branches are short with many splits during peak outbreak in the second half of 2014 but become mostly long in late 2014 and the remainder of 2015.  According to the phylogenetic Brownian process model outlined in Section \ref{sec:phyloBrown}, virus-specific rate values are more highly correlated to one another when closely located to one-another on the the phylogenetic tree.  It is plausible that these correlations allow the phylogenetic Hawkes model to infer higher rates for some strains that survive late into the epidemic despite dropping case counts.  The model attributes some of the highest values to strains appearing in Coyah, Conakry and Kindia, Guinea, in late 2014 and early 2015.  Interestingly, the model also attributes its lowest values to cases in Kono, Sierra Leone, in early 2015.  Taken together, Figure \ref{fig:tree} and the significant virus lists of Appendix \ref{sec:signifTaxa} may provide helpful leads for epidemiologists searching for variants with heightened relative rates of contagion.

\section{Discussion}

We propose the phylogenetic Hawkes process, a Bayesian hierarchical model that relates viral spatial contagion to molecular evolution by uniting the two epidemiological paradigms of self-exciting point process and phylogenetic modeling.  Due to difficulties in scaling the model to larger numbers of observations, we advance a computing strategy that combines HMC, dynamic programming and massive parallelization for key inferential bottlenecks.  Finally, we apply our novel model and high performance computing framework to the analysis of over 23,000 viral cases arising from the 2014-2016 Ebola outbreak in West Africa, and Ebola strains and subtrees with plausibly higher degrees of contagiousness reveal themselves.

Unfortunately, the current model will fail when applied to the analysis of a global pandemic due to the dominant role of non-local, large scale transportation networks in propagating viral spread \citep{holbrook2020massive}.   We are particularly interested in developing extensions to the phylogenetic Hawkes process that leverage recent advances in scaling high-dimensional multivariate Hawkes processes \citep{nickel2020learning} and applying the resulting multivariate phylogenetic Hawkes process to the analysis of global pandemics.  In this context, each additional dimension will represent an additional country or province.   Prodigious computational challenges are inevitable, and we suspect that non-trivial GPU implementations will be necessary for big data applications.

Moving beyond inference, a major question is whether the phylogenetic Hawkes process can be useful for prediction of spatial contagion and dynamics.  Here, recent neural network extensions of the Hawkes process might prove useful \citep{mei2017neural,zuo2020transformer,zhang2020self}, but it is unclear what forward simulation of phylogenetic branching dynamics would look like in the context of a Hawkes process.  Moreover, generating samples from the posterior predictive distribution of a Hawkes process would be extremely time consuming when one is conditioning on millions of posterior samples.  To work around this, one could perhaps parallelize over fixed parameter settings the algorithm of \citet{dassios2011dynamic} for simulating Hawkes processes when the temporal triggering function is exponential.  Such an implementation would require efficient use of parallel pseudo-random number generators \citep{salmon2011parallel}.

\section*{Acknowledgments}

The research leading to these results has received funding through National Institutes of Health grants K25 AI153816, U19 AI135995 and R01 AI153044 and National Science Foundation grant DMS1264153.
We gratefully acknowledge support from NVIDIA Corporation and Advanced Micro Devices, Inc.~with the donation of parallel computing resources used for this research.

\appendix

\section{The likelihood integral with event-specific rates}\label{sec:integral}

Without loss of generality, we consider the temporal Hawkes process with constant background rate.
To compute the likelihood (Equation \ref{eq:likelihood}), we must calculate the integral
\begin{align*}
\Lambda(t_{N}) &= \int_0^{t_N} \lambda(t)\, \mbox{\emph{d}} t=\int_0^{t_{N}} \left(\mu +  \theta_0 \sum_{t_n<t} \theta_n \omega e^{- \omega\, (t-t_n) }\right) \dd t \\ &= \int_0^{t_{1}} \mu\, \dd t + \sum_{n=1}^{N-1} \int_{t_n}^{t_{n+1}} \left(\mu +  \theta_0 \sum_{t_n<t} \theta_n \omega e^{- \omega\, (t-t_n) }\right) \dd t \\
&= \mu \,t_{N} + \theta_0\omega \sum_{n=1}^{N-1} \int_{t_n}^{t_{n+1}}  \sum_{n'=1}^n \theta_{n'} e^{- \omega\, (t-t_{n'}) }\, \dd t \\
&=  \mu \,t_{N} + \theta_0\omega \sum_{n=1}^{N-1} \sum_{n'=1}^n \theta_{n'} \int_{t_n}^{t_{n+1}}   e^{- \omega\, (t-t_{n'}) }\, \dd t \\
&= \mu \,t_{N} -  \theta_0 \sum_{n=1}^{N-1} \sum_{n'=1}^n \theta_{n'}  \left(  e^{- \omega\, (t_{n+1}-t_{n'}) }- e^{- \omega\, (t_{n}-t_{n'}) } \right) \, ,
\end{align*}
and we further simplify the double summation in the following.

\begin{claim}
	The temporal Hawkes process with rate function
	\begin{align}
	\lambda(t) = \mu +  \theta_0 \sum_{t_n<t} \theta_n \omega e^{- \omega\, (t-t_n) }
	\end{align}
	admits the integral
	\begin{align}
	\Lambda(t_N) = \int_0^{t_N} \lambda(t)\, \mbox{\emph{d}} t = \mu\, t_N - \theta_0 \sum_{n=1}^{N-1} \theta_n \left(e^{-\omega(t_N-t_n)} -1 \right) \, .
	\end{align}
\end{claim}
\begin{proof}
	Proceeding by induction, the assertion is trivial for $N=1$.  If it is true for some $N>0$, this implies that
	\begin{align*}
	\sum_{n=1}^{N-1} \sum_{n'=1}^n \theta_{n'}  \left(  e^{- \omega\, (t_{n+1}-t_{n'}) }- e^{- \omega\, (t_{n}-t_{n'}) } \right) = \sum_{n=1}^{N-1} \theta_n \left(e^{-\omega(t_N-t_n)} -1 \right)\, .
	\end{align*}
	It follows that
	\begin{align*}
	\Lambda(t_{N+1}) &= \mu \,t_{N+1} -  \theta_0 \sum_{n=1}^{N} \sum_{n'=1}^n \theta_{n'}  \left(  e^{- \omega\, (t_{n+1}-t_{n'}) }- e^{- \omega\, (t_{n}-t_{n'}) } \right) \\
	&=    \mu \,(t_{N+1}-t_N) -\theta_0 \sum_{n'=1}^N \theta_{n'}  \left(  e^{- \omega\, (t_{N+1}-t_{n'}) }- e^{- \omega\, (t_{N}-t_{n'}) } \right) \\
	& \quad + \mu \,t_N  - \theta_0 \sum_{n=1}^{N-1} \sum_{n'=1}^n \theta_{n'}  \left(  e^{- \omega\, (t_{n+1}-t_{n'}) }- e^{- \omega\, (t_{n}-t_{n'}) } \right) \\
	&=  \mu \,(t_{N+1}-t_N) -\theta_0 \sum_{n'=1}^N \theta_{n'}  \left(  e^{- \omega\, (t_{N+1}-t_{n'}) }- e^{- \omega\, (t_{N}-t_{n'}) } \right) + \Lambda(t_n) \\
	&=    \mu \,(t_{N+1}-t_N) -\theta_0 \sum_{n'=1}^N \theta_{n'}  \left(  e^{- \omega\, (t_{N+1}-t_{n'}) }- e^{- \omega\, (t_{N}-t_{n'}) } \right) \\
	& \quad + \mu \,t_N  -  \theta_0 \sum_{n=1}^{N-1} \theta_n \left(e^{-\omega(t_N-t_n)} -1 \right) \\
	&= \mu \, t_{N+1}  -\theta_0 \sum_{n=1}^N \theta_{n}  \left(  e^{- \omega\, (t_{N+1}-t_{n}) }- e^{- \omega\, (t_{N}-t_{n}) } \right) \\
	& \quad   -  \theta_0 \sum_{n=1}^{N-1} \theta_n \left(e^{-\omega(t_N-t_n)} -1 \right) \\
	&= \mu \, t_{N+1}  - \theta_0 \sum_{n=1}^N \theta_n \left(e^{- \omega\, (t_{N+1}-t_{n}) } - 1 \right) \, ,
	\end{align*}
	thus completing the proof.
\end{proof}

\section{Viruses with statistically significant relative rates}\label{sec:signifTaxa}

The following observed viruses have relative rates for which the 95\% credible intervals have lower bound greater than 1.
\begin{multicols}{2}
	\begin{Verbatim}[numbers=left,fontsize=\tiny]
	EBOV|EM_COY_2015_013857||GIN|Forecariah|2015-03-18
	EBOV|EM_COY_2015_013731||GIN|Coyah|2015-03-14
	EBOV|EM_COY_2015_014261||GIN|Forecariah|2015-04-02
	EBOV|IPDPFHGINSP_GUI_2015_5339||GIN|Conakry|2015-04-08
	EBOV|KG12||GIN|Boke|2015-05-27
	EBOV|EM_COY_2015_017091||GIN|Boke|2015-05-28
	EBOV|EM_COY_2015_016483||GIN|Boke|2015-05-13
	EBOV|EM_COY_2015_016743||GIN|Boke|2015-05-19
	EBOV|KG88||GIN|Boke|2015-06-19
	EBOV|KG90||GIN|Boke|2015-06-19
	EBOV|KG87||GIN|Boke|2015-06-19
	EBOV|KG80||GIN|Boke|2015-06-18
	EBOV|KG35||GIN|Boke|2015-06-08
	EBOV|GUI_CTS_2015_0052||GIN|Boke|2015-06-25
	EBOV|GUI_CTS_2015_0050||GIN|Boke|2015-06-20
	EBOV|GUI_CTS_2015_0051||GIN|Boke|2015-06-21
	EBOV|KG91||GIN|Boke|2015-06-20
	EBOV|KG45||GIN|Boke|2015-06-09
	EBOV|EM_COY_2015_017574||GIN|Boke|2015-06-10
	EBOV|IPDPFHGINSP_GUI_2015_6899||GIN|Boke|2015-05-14
	EBOV|EM_COY_2015_014102||GIN|Conakry|2015-03-26
	EBOV|IPDPFHGINSP_GUI_2015_4909||GIN|Conakry|2015-03-29
	EBOV|IPDPFHGINSP_GUI_2015_5117||GIN|Dubreka|2015-04-03
	EBOV|IPDPFHGINSP_GUI_2015_4786||GIN|Conakry|2015-03-26
	EBOV|EM_COY_2015_014098||GIN|Conakry|2015-03-26
	EBOV|EM_COY_2015_014100||GIN|Conakry|2015-03-26
	EBOV|1648|KR534569|GIN|Kindia|2014-10-23
	EBOV|20140174|KR653294|SLE|WesternUrban|2014-08-27
	EBOV|20144521|KR653271|SLE|Kono|2014-11-12
	EBOV|DML24511|KT357828|SLE|Kono|2015-01-14
	EBOV|20146001|KR653237|SLE|Kono|2014-11-24
	EBOV|20141352|KR653284|SLE|Kambia|2014-09-26
	EBOV|G6069.1|KR105344|SLE|Kambia|2014-09-25
	EBOV|20141997|KR653293|SLE|Kono|2014-10-10
	EBOV|J0008|KP759718|SLE|WesternUrban|2014-09-29
	EBOV|KT2449|KU296465|SLE|WesternUrban|2014-12-05
	EBOV|20143458|KR653245|SLE|Tonkolili|2014-11-01
	EBOV|J0016|KP759742|SLE|WesternRural|2014-09-30
	EBOV|G6104.1|KR105349|SLE|Moyamba|2014-09-28
	EBOV|20141397|KR653288|SLE|Moyamba|2014-09-28
	EBOV|G6103.1|KR105348|SLE|Moyamba|2014-09-28
	EBOV|PL6656|KU296523|SLE|PortLoko|2015-05-01
	EBOV|PL7429|KU296617|SLE|PortLoko|2015-05-29
	EBOV|PL7451|KU296837|SLE|PortLoko|2015-05-29
	EBOV|PL7375|KU296498|SLE|PortLoko|2015-05-27
	EBOV|PL7401|KU296819|SLE|PortLoko|2015-05-27
	EBOV|G4725.1|KR105267|SLE|Moyamba|2014-08-04
	EBOV|J0038|KP759618|SLE|PortLoko|2014-09-30
	EBOV|J0102|KP759657|SLE|PortLoko|2014-10-27
	EBOV|J0103|KP759658|SLE|PortLoko|2014-10-29
	EBOV|J0104|KP759765|SLE|PortLoko|2014-10-29
	EBOV|J0105|KP759766|SLE|PortLoko|2014-10-29
	EBOV|20141282|KR653229|SLE|Kambia|2014-09-23
	EBOV|J0027|KP759678|SLE|Bombali|2014-10-02
	EBOV|G6091.1|KR105346|SLE|Tonkolili|2014-09-27
	EBOV|J0162|KP759622|SLE|Kambia|2014-11-09
	EBOV|J0034|KP759608|SLE|PortLoko|2014-09-30
	EBOV|20141429|KR653260|SLE|Kono|2014-09-28
	EBOV|G6089.1|KR105345|SLE|Tonkolili|2014-09-27
	EBOV|G6095.1|KR105347|SLE|Tonkolili|2014-09-27
	EBOV|20142551|KR653226|SLE|Koinadugu|2014-10-23
	EBOV|20143648|KR653253|SLE|Koinadugu|2014-11-03
	EBOV|20143938|KR653264|SLE|Koinadugu|2014-11-07
	EBOV|20143659|KR653256|SLE|Koinadugu|2014-11-03
	EBOV|20142895|KR653268|SLE|Koinadugu|2014-10-24
	EBOV|20141271|KR653261|SLE|PortLoko|2014-09-24
	EBOV|J0084|KP759753|SLE|PortLoko|2014-10-16
	EBOV|J0083|KP759752|SLE|PortLoko|2014-10-18
	EBOV|J0035|KP759609|SLE|PortLoko|2014-10-01
	EBOV|J0058|KP759724|SLE|PortLoko|2014-10-08
	EBOV|J0071|KP759736|SLE|PortLoko|2014-10-09
	EBOV|20140024|KR653252|SLE|PortLoko|2014-08-20
	EBOV|G4972.1|KR105285|SLE|Kenema|2014-08-14
	EBOV|20140433|KR653246|SLE|Tonkolili|2014-09-03
	EBOV|J0030|KP759606|SLE|PortLoko|2014-10-03
	EBOV|J0011|KP759641|SLE|PortLoko|2014-09-29
	EBOV|20141288|KR653232|SLE|PortLoko|2014-09-23
	EBOV|J0013|KP759643|SLE|PortLoko|2014-09-28
	EBOV|J0014|KP759740|SLE|PortLoko|2014-09-29
	EBOV|J0114|KP759596|SLE|PortLoko|2014-10-29
	EBOV|J0115|KP759597|SLE|PortLoko|2014-10-29
	EBOV|G4956.1|KR105282|SLE|Tonkolili|2014-08-13
	EBOV|20140134|KR653227|SLE|Bombali|2014-08-26
	EBOV|J0020|KP759755|SLE|Bombali|2014-09-25
	EBOV|J0019|KP759754|SLE|Bombali|2014-09-25
	EBOV|J0022|KP759757|SLE|Bombali|2014-09-25
	EBOV|J0017|KP759747|SLE|Bombali|2014-09-26
	EBOV|20142260|KR653262|SLE|Bombali|2014-10-14
	EBOV|J0096|KP759654|SLE|WesternUrban|2014-10-27
	EBOV|J0015|KP759741|SLE|Bombali|2014-09-26
	EBOV|20141232|KR653266|SLE|Tonkolili|2014-09-25
	EBOV|20142127|KR653234|SLE|Bombali|2014-10-12
	EBOV|20141497|KR653273|SLE|Bombali|2014-10-01
	EBOV|G5743.1|KR105323|SLE|Kono|2014-09-17
	EBOV|G5529.1|KR105309|SLE|Kono|2014-09-05
	EBOV|12854_EMLH|KU296729|SLE|WesternUrban|2015-04-15
	EBOV|13031_EMLH|KU296665|SLE|WesternUrban|2015-05-18
	EBOV|DML13828||SLE|WesternUrban|2015-06-22
	EBOV|20143716|KR653292|SLE|Moyamba|2014-11-04
	EBOV|20144610|KR653270|SLE|Tonkolili|2014-11-12
	EBOV|J0021|KP759756|SLE|Bombali|2014-09-29
	EBOV|PL5260|KU296591|SLE|Kambia|2015-03-23
	EBOV|EM_COY_2015_013962||GIN|Forecariah|2015-03-22
	EBOV|EM_COY_2015_013671||GIN|Coyah|2015-03-12
	EBOV|PL4347|KU296720|SLE|Kambia|2015-03-04
	EBOV|EM_COY_2015_015815||GIN|Fria|2015-04-14
	EBOV|IPDPFHGINSP_GUI_2015_5038||GIN|Conakry|2015-04-01
	EBOV|EM_COY_2015_016414||GIN|Dubreka|2015-05-11
	EBOV|EM_COY_2015_017865||GIN|Dubreka|2015-06-18
	EBOV|REDC-GUI-2015-00502||GIN|Conakry|2015-07-13
	EBOV|REDC_GUI_2015_00483||GIN|Conakry|2015-07-12
	EBOV|REDC-GUI-2015-00402||GIN|Conakry|2015-07-08
	EBOV|EM_COY_2015_017057||GIN|Dubreka|2015-05-27
	EBOV|EM_COY_2015_017135||GIN|Dubreka|2015-05-29
	EBOV|EM_COY_2015_016800||GIN|Dubreka|2015-05-20
	EBOV|EM_COY_2015_016617||GIN|Dubreka|2015-05-16
	EBOV|EM_COY_2015_016531||GIN|Dubreka|2015-05-14
	EBOV|PL4094|KU296476|SLE|Kambia|2015-02-27
	EBOV|14795_EMLK|KU296737|SLE|Kambia|2015-04-25
	EBOV|EM_COY_2015_014370||GIN|Forecariah|2015-04-07
	EBOV|EM_COY_2015_013795||GIN|Forecariah|2015-03-16
	EBOV|LIBR10106|KT725337|LBR|Nimba|2014-08-20
	EBOV|LIBR10081|KT725339|LBR|Margibi|2014-08-26
	EBOV|LIBR10017|KT725293|LBR|Montserrado|2014-08-19
	EBOV|LIBR10170|KT725350|LBR|Montserrado|2014-08-17
	EBOV|CDC-NIH-257||LBR|Montserrado|2014-08-27
	EBOV|LIBR10082|KT725300|LBR|Margibi|2014-08-26
	EBOV|LIBR10029|KT725348|LBR|Margibi|2014-08-16
	EBOV|LIBR10030|KT725257|LBR|Margibi|2014-08-16
	EBOV|LIBR10025|KT725315|LBR|Margibi|2014-08-16
	EBOV|LIBR10088|KT725349|LBR|Montserrado|2014-08-26
	EBOV|LIBR10087|KT725346|LBR|Margibi|2014-08-26
	EBOV|LIBR10086|KT725381|LBR|Margibi|2014-08-26
	EBOV|EM_075368|KR817131|GIN|Macenta|2014-08-29
	EBOV|EM_000015|KR817067|GIN|Macenta|2014-09-01
	EBOV|EM_000321|KR817074|GIN|Macenta|2014-09-09
	EBOV|EM_074785|KR817126|GIN|Macenta|2014-08-16
	EBOV|EM_000218|KR817072|GIN|Macenta|2014-09-07
	EBOV|EM_000127|KR817070|GIN|Gueckedou|2014-09-04
	EBOV|EM_000028|KR817069|GIN|Macenta|2014-09-01
	EBOV|EM_000502|KR817078|GIN|Macenta|2014-09-13
	EBOV|EM_001102|KR817089|GIN|Lola|2014-10-05
	EBOV|Conakry_645|KR534513|GIN|Macenta|2014-08-14
	EBOV|EM_074548|KR817123|LBR|Lofa|2014-08-08
	EBOV|EM_078416|KR817161|GIN|Faranah|2014-11-24
	EBOV|EM_078415|KR817160|GIN|Faranah|2014-11-24
	EBOV|EM_075447|KR817134|GIN|Macenta|2014-08-31
	EBOV|EM_075373|KR817132|GIN|Macenta|2014-08-30
	EBOV|789|KR534523|GIN|Dubreka|2014-08-29
	EBOV|1622|KR534567|GIN|Nzerekore|2014-10-22
	EBOV|EM_075435|KR817133|GIN|Kerouane|2014-08-30
	EBOV|EM_004201|KR817093|GIN|Kissidougou|2014-12-24
	EBOV|EM_004259|KR817094|GIN|Kissidougou|2014-12-26
	EBOV|EM_004290|KR817095|GIN|Kissidougou|2014-12-27
	EBOV|EM_078722|KR817175|GIN|Kissidougou|2014-12-16
	EBOV|EM_078779|KR817177|GIN|Kissidougou|2014-12-18
	EBOV|EM_078694|KR817171|GIN|Kissidougou|2014-12-14
	EBOV|EM_004059|KR817090|GIN|Kissidougou|2014-12-20
	EBOV|EM_078706|KR817173|GIN|Kissidougou|2014-12-15
	EBOV|EM_004085|KR817091|GIN|Kissidougou|2014-12-19
	EBOV|EM_078763|KR817176|GIN|Kissidougou|2014-12-17
	EBOV|LIBR0284|KT725328|LBR|GrandCapeMount|2014-11-22
	EBOV|CDC-NIH-3832||LBR|GrandCapeMount|2014-11-23
	EBOV|LIBR0286|KR006952|LBR|GrandCapeMount|2014-11-22
	EBOV|LIBR10103|KT725338|LBR|Nimba|2014-08-20
	EBOV|LIBR10035|KT725259|LBR|GrandBassa|2014-08-15
	EBOV|EM_074438|KR817118|GIN|Nzerekore|2014-08-01
	EBOV|EM_000706|KR817079|GIN|Nzerekore|2014-09-21
	EBOV|EM_074439|KR817119|GIN|Nzerekore|2014-08-01
	EBOV|EM_074436|KR817116|GIN|Nzerekore|2014-08-01
	EBOV|EM_074437|KR817117|GIN|Nzerekore|2014-08-01
	EBOV|LIBR10127|KT725272|LBR|Bong|2014-08-22
	EBOV|LIBR0071|KT725303|LBR|GrandBassa|2014-11-06
	EBOV|LIBR0059|KR006941|LBR|RiverCess|2014-11-05
	EBOV|LIBR0058|KR006940|LBR|RiverCess|2014-11-05
	EBOV|EM095|KM034550|SLE|Kailahun|2014-05-25
	EBOV|G3676|KM034554|SLE|Kailahun|2014-05-27
	\end{Verbatim}
\end{multicols}
\noindent
The following observed viruses have relative rates for which the 95\% credible intervals have upper bound less than 1.
\begin{Verbatim}[numbers=left,fontsize=\tiny]
EBOV|EM_079413|KR817184|GIN|Gueckedou|2014-03-31
EBOV|EM_079414|KR817185|GIN|Gueckedou|2014-03-31
EBOV|EM_079542|KR817199|GIN|Gueckedou|2014-04-12
EBOV|EM_079578|KR817201|GIN|Gueckedou|2014-04-18
EBOV|EM_079587|KR817202|GIN|Gueckedou|2014-04-22
EBOV|EM_079514|KR817197|GIN|Gueckedou|2014-04-10
\end{Verbatim}

\section{Parallelized gradient algorithms}\label{sec:parallel}

Algorithms \ref{alg:lik2} and \ref{alg:lik} present instructions for computing the Hawkes process log-likelihood gradient (Equation \ref{eq:gradient}) with respect to the $M$-vector $\ttheta$ of event-specific rates.  Figure \ref{fig:speedups} shows resulting speedups for both Algorithm \eqref{alg:lik2} and Algorithm \eqref{alg:lik} on a CPU and GPU, respectively.

\newcommand{\bb}{\mathbf{b}}
\newcommand{\vv}{\mathbf{v}}
\newcommand{\aaa}{\mathbf{a}}
\newcommand{\adapt}{\mathbf{l}}

\newcounter{algsubstate}
\renewcommand{\thealgsubstate}{\alph{algsubstate}}
\newenvironment{algsubstates}
{\setcounter{algsubstate}{0}%
	\renewcommand{\State}{%
		\stepcounter{algsubstate}%
		\Statex {\footnotesize\thealgsubstate:}\space}}
{}

\newcommand{\transformR}{r}
\newcommand{\transformCDF}{c}
\newcommand{\threadsPerBlock}{B}

\algblock{ParFor}{EndParFor}
\algnewcommand\algorithmicparfor{\textbf{parfor}}
\algnewcommand\algorithmicpardo{\textbf{do}}
\algnewcommand\algorithmicendparfor{\textbf{end\ parfor}}
\algrenewtext{ParFor}[1]{\algorithmicparfor\ #1\ \algorithmicpardo}
\algrenewtext{EndParFor}{\algorithmicendparfor}

\newcommand{\blockSize}{B}

\newcommand{\llambda}{\boldsymbol{\lambda}}
\newcommand{\Ddelta}{\boldsymbol{\Delta}}

\begin{singlespace}
\begin{algorithm}
	\scriptsize
	\caption{Parallel evaluation of Hawkes process log-likelihood gradient: \\
		\emph{Makes use of multiple central processing unit (CPU) cores and loop vectorization to calculate Hawkes process log-likelihood gradient with respect to event-specific rates $\ttheta$.  When using double-precision floating point, this algorithm may use either SSE or AVX vectorization to make $J=2$- or $4$-long jumps. We denote the number of CPU cores as $B$.   Symbols $\ell$, $\lambda$ and $\Lambda$ appear in Equations \eqref{eq:likelihood} and \eqref{eq:gradient}.}
	}\label{alg:lik2}
	\begin{algorithmic}[1]
		\State Compute rates $\lambda_1, \dots, \lambda_N$:
		\begin{algsubstates}
			\State \hspace{1em} \textbf{parfor} $b \in \{1,\dots,B\}$ \textbf{do}
			\State \hspace{2em} \textbf{if} $b\neq B$ \textbf{then}
			\State \hspace{3em}$Upper \gets b \lfloor  N/B  \rfloor$
			\State \hspace{2em} \textbf{else}
			\State \hspace{3em} $Upper \gets  \lceil  N/B  \rceil$
			\State \hspace{2em} \textbf{end if}
			\State \hspace{2em} \textbf{for} $n \in  \{ (b-1)\lfloor N/B  \rfloor + 1,  \dots,Upper \}$ \textbf{do}
			\State \hspace{3em} copy $\x_{n}$, $t_{n}$ to cache
			\State \hspace{3em} $\llambda_{n} \gets \mathbf{0}$ \Comment{vector of length J}
			\State  \hspace{3em} $n' \gets 1$
			\State  \hspace{3em} \textbf{while} $n' < N$ \textbf{do}
			\State \hspace{4em} $J \gets \min(J,N-n')$
			\State \hspace{4em} copy $\x_{n':(n'+J)}$, $t_{n':(n'+J)}$ to cache
			\State \hspace{4em} $\Ddelta_{nn'}:\Ddelta_{nn':(n'+J-1)} \gets (\x_{n} - \x_{n'}):(\x_{n} - \x_{n'+J-1})$ \Comment{vectorized subtraction}
			\State \hspace{4em} calculate $\delta_{nn'}:\delta_{n(n'+J-1)}$  \Comment{vectorized multiplication}
			\State \hspace{4em} calculate $\lambda_{nn'}:\lambda_{n(n'+J-1)}$  \Comment{vectorized evaluation}
			\State \hspace{4em} $\llambda_{n} \gets \llambda_{n} +  \lambda_{nn'}:\lambda_{n(n'+J-1)}$		\Comment{vectorized addition}
			\State \hspace{4em} $n' \gets n' + J$
			\State \hspace{3em} \textbf{end while}
			\State \hspace{2em} \textbf{end for}
			\State \hspace{1em} \textbf{end parfor}
		\end{algsubstates}

		\State Compute $M$ gradients $\frac{\partial \ell}{\partial \theta_n}$:
		\begin{algsubstates}
			\State \hspace{1em} \textbf{parfor} $b \in \{1,\dots,B\}$ \textbf{do}
			\State \hspace{2em} \textbf{if} $b\neq B$ \textbf{then}
			\State \hspace{3em}$Upper \gets b \lfloor  M/B  \rfloor$
			\State \hspace{2em} \textbf{else}
			\State \hspace{3em} $Upper \gets  \lceil  M/B  \rceil$
			\State \hspace{2em} \textbf{end if}
			\State \hspace{2em} \textbf{for} $n \in  \{ (b-1)\lfloor M/B  \rfloor + 1,  \dots,Upper \}$ \textbf{do}
			\State \hspace{3em} copy $\x_{n}$, $t_{n}$ to cache
			\State \hspace{3em} $\frac{\partial \ell}{\partial \theta_n} \gets 0$
			\State  \hspace{3em} $n' \gets 1$
			\State  \hspace{3em} \textbf{while} $n' < N$ \textbf{do}
			\State \hspace{4em} $J \gets \min(J,N-n')$
			\State \hspace{4em} copy $\x_{n':(n'+J)}$, $t_{n':(n'+J)}$ to cache
			\State \hspace{4em} $\Ddelta_{nn'}:\Ddelta_{nn':(n'+J-1)} \gets (\x_{n} - \x_{n'}):(\x_{n} - \x_{n'+J-1})$ \Comment{vectorized subtraction}
			\State \hspace{4em} calculate $\delta_{nn'}:\delta_{n(n'+J-1)}$  \Comment{vectorized multiplication}
			\State \hspace{4em} calculate $e^{- \omega\, (t_{n'}-t_n) }  \phi\left(\frac{\delta_{nn'}}{h}\right):e^{- \omega\, (t_{n'+J-1}-t_n) }  \phi\left(\frac{\delta_{n(n'+J-1)}}{h}\right)$  \Comment{vectorized evaluation}
			\State \hspace{4em} \textbf{for} $j \in {n',\dots,n'+J-1}$ \textbf{do}
			\State \hspace{5em} $\frac{\partial \ell}{\partial \theta_n} \gets \frac{\partial \ell}{\partial \theta_n}  + \mathcal{I}_{[t_n<t_{j}]} \frac{1}{\lambda_{j}}  \frac{\partial \lambda_{jn}}{\partial \theta_n} $
			\State \hspace{4em} \textbf{end for}
			\State \hspace{4em} $n' \gets n' + J$
			\State \hspace{3em} \textbf{end while}
			\State \hspace{2em} $\frac{\partial \ell}{\partial \theta_n} \gets \frac{\partial \ell}{\partial \theta_n} +  \theta_0  \left(e^{- \omega\, (t_{N}-t_{n}) } - 1 \right)$
			\State \hspace{2em} \textbf{end for}
			\State \hspace{1em} \textbf{end parfor}
		\end{algsubstates}
	\end{algorithmic}
\end{algorithm}

\end{singlespace}

\begin{algorithm}
	\scriptsize
	\caption{Parallel evaluation of Hawkes process log-likelihood gradient: \\
		\emph{Computes the log-likelihood gradient with respect to event-specific rates $\ttheta$ using multiple levels of parallelization on a graphics processing unit (GPU).   In this paper, we specify $B=128$ for the size of the GPU work groups.  Symbols $\ell$, $\lambda$ and $\Lambda$ appear in Equations \eqref{eq:likelihood} and \eqref{eq:gradient}.}
	}\label{alg:lik}
	\begin{algorithmic}[1]
		\State Compute rates $\lambda_1, \dots, \lambda_N$:
		\begin{algsubstates}
			\State	\hspace{1em}	\textbf{parfor} $n \in \{1,\dots,N\}$ \textbf{do}
			\State\hspace{2em} copy $\x_n$, $t_n$ to local \Comment{$B$ threads}
			\State\hspace{2em}	\textbf{parfor} $N'\in \{1,\dots,\lfloor N/B\rfloor\}$ \textbf{do}
			\State\hspace{3em} $n' \gets N'$
			\State\hspace{3em} $\lambda_{nN'} \gets 0$
			\State\hspace{3em}\textbf{while} $n' < N$ \textbf{do}
			\State\hspace{4em} copy $\x_{n'}$, $t_{n'}$ to local \Comment{$B$ threads}
			\State\hspace{4em} $\Ddelta_{nn'} \gets \x_n - \x_{n'}$ \Comment{vectorized subtraction}
			\State\hspace{4em} calculate $\delta_{nn'} = \sqrt{\sum \Ddelta_{nn'}\circ \Ddelta_{nn'}}$  \Comment{vectorized multiplication}
			\State\hspace{4em} $\lambda_{nN'} \gets \lambda_{nN'}  +\lambda_{nn'}$ \Comment{$\lambda_{nn'}$ a function of $\delta_{nn'}$, $t_n$ and $t_{n'}$}
			\State\hspace{4em} $n' \gets n' + B$
			\State\hspace{3em}	\textbf{end while}
			\State\hspace{2em}\textbf{end parfor}
			\State\hspace{2em} $\lambda_n\gets \sum_{N'} \lambda_{nN'} $   \Comment{binary tree reduction on chip}
			\State\hspace{1em}	\textbf{end parfor}
		\end{algsubstates}

		\State Compute $M$ gradients $\frac{\partial \ell}{\partial \theta_n}$:
		\begin{algsubstates}
			\State	\hspace{1em}	\textbf{parfor} $n \in \{1,\dots,M\}$ \textbf{do}
			\State\hspace{2em} copy $\x_n$, $t_n$ to local \Comment{$B$ threads}
			\State\hspace{2em}	\textbf{parfor} $N'\in \{1,\dots,\lfloor N/B\rfloor\}$ \textbf{do}
			\State\hspace{3em} $n' \gets N'$
			\State\hspace{3em} $\left(\frac{\partial \ell}{\partial \theta_n}\right)_{N'} \gets 0$
			\State\hspace{3em}\textbf{while} $n' < N$ \textbf{do}
			\State\hspace{4em} copy $\x_{n'}$, $t_{n'}$ to local \Comment{$B$ threads}
			\State\hspace{4em} $\Ddelta_{nn'} \gets \x_n - \x_{n'}$ \Comment{vectorized subtraction}
			\State\hspace{4em} calculate $\delta_{nn'} = \sqrt{\sum \Ddelta_{nn'}\circ \Ddelta_{nn'}}$  \Comment{vectorized multiplication}
			\State\hspace{4em} $\left(\frac{\partial \ell}{\partial \theta_n}\right)_{N'} \gets \left(\frac{\partial \ell}{\partial \theta_n}\right)_{N'}  + \mathcal{I}_{[t_n<t_{n'}]} \frac{1}{\lambda_{n'}}  \frac{\partial \lambda_{n'n}}{\partial \theta_n} $
			\State\hspace{4em} $n' \gets n' + B$
			\State\hspace{3em}	\textbf{end while}
			\State\hspace{2em}\textbf{end parfor}
			\State\hspace{2em} $\frac{\partial \ell}{\partial \theta_n}\gets \sum_{N'}\left(\frac{\partial \ell}{\partial \theta_n}\right)_{N'} $   \Comment{binary tree reduction on chip}
			\State\hspace{2em} $\frac{\partial \ell}{\partial \theta_n}\gets \frac{\partial \ell}{\partial \theta_n} + \theta_0  \left(e^{- \omega\, (t_{N}-t_{n}) } - 1 \right)$
			\State\hspace{1em}	\textbf{end parfor}
		\end{algsubstates}

	\end{algorithmic}
\end{algorithm}

\bibliographystyle{sysbio}
\bibliography{refs}

\begin{thebibliography}{79}
\providecommand{\natexlab}[1]{#1}
\providecommand{\selectlanguage}[1]{\relax}
\providecommand{\bibAnnoteFile}[1]{%
  \IfFileExists{#1}{\begin{quotation}\noindent\textsc{Key:} #1\\
  \textsc{Annotation:}\ \input{#1}\end{quotation}}{}}
\providecommand{\bibAnnote}[2]{%
  \begin{quotation}\noindent\textsc{Key:} #1\\
  \textsc{Annotation:}\ #2\end{quotation}}

\bibitem[{Bacry et~al.(2015)Bacry, Mastromatteo, and Muzy}]{bacry2015hawkes}
Bacry, E., I.~Mastromatteo, and J.-F. Muzy. 2015. Hawkes processes in finance.
  Market Microstructure and Liquidity 1:1550005.
\bibAnnoteFile{bacry2015hawkes}

\bibitem[{Baize et~al.(2014)Baize, Pannetier, Oestereich, Rieger, Koivogui,
  Magassouba, Soropogui, Sow, Ke{\"\i}ta, De~Clerck
  et~al.}]{baize2014emergence}
Baize, S., D.~Pannetier, L.~Oestereich, T.~Rieger, L.~Koivogui, N.~Magassouba,
  B.~Soropogui, M.~S. Sow, S.~Ke{\"\i}ta, H.~De~Clerck, et~al. 2014. Emergence
  of zaire ebola virus disease in guinea. New England Journal of Medicine
  371:1418--1425.
\bibAnnoteFile{baize2014emergence}

\bibitem[{Bertozzi et~al.(2020)Bertozzi, Franco, Mohler, Short, and
  Sledge}]{bertozzi2020challenges}
Bertozzi, A.~L., E.~Franco, G.~Mohler, M.~B. Short, and D.~Sledge. 2020. The
  challenges of modeling and forecasting the spread of covid-19. Proceedings of
  the National Academy of Sciences 117:16732--16738.
\bibAnnoteFile{bertozzi2020challenges}

\bibitem[{Boni et~al.(2020)Boni, Lemey, Jiang, Lam, Perry, Castoe, Rambaut, and
  Robertson}]{boni2020evolutionary}
Boni, M.~F., P.~Lemey, X.~Jiang, T.~T.-Y. Lam, B.~W. Perry, T.~A. Castoe,
  A.~Rambaut, and D.~L. Robertson. 2020. Evolutionary origins of the sars-cov-2
  sarbecovirus lineage responsible for the covid-19 pandemic. Nature
  Microbiology 5:1408--1417.
\bibAnnoteFile{boni2020evolutionary}

\bibitem[{Brockmann and Helbing(2013)}]{brockmann2013hidden}
Brockmann, D. and D.~Helbing. 2013. The hidden geometry of complex,
  network-driven contagion phenomena. science 342:1337--1342.
\bibAnnoteFile{brockmann2013hidden}

\bibitem[{Cavalli-Sforza and Edwards(1967)}]{cavalli1967phylogenetic}
Cavalli-Sforza, L.~L. and A.~W. Edwards. 1967. Phylogenetic analysis. models
  and estimation procedures. American Journal of Human Genetics 19:233--257.
\bibAnnoteFile{cavalli1967phylogenetic}

\bibitem[{Chiang et~al.(2020)Chiang, Liu, and Mohler}]{chiang2020hawkes}
Chiang, W.-H., X.~Liu, and G.~Mohler. 2020. Hawkes process modeling of covid-19
  with mobility leading indicators and spatial covariates. medRxiv .
\bibAnnoteFile{chiang2020hawkes}

\bibitem[{Cybis et~al.(2015)Cybis, Sinsheimer, Bedford, Mather, Lemey, and
  Suchard}]{cybis2015assessing}
Cybis, G., J.~Sinsheimer, T.~Bedford, A.~Mather, P.~Lemey, and M.~Suchard.
  2015. Assessing phenotypic correlation through the multivariate phylogenetic
  latent liability model. Annals of Applied Statistics 9:969 -- 991.
\bibAnnoteFile{cybis2015assessing}

\bibitem[{Daley and Jones(2003)}]{daley2003introduction}
Daley, D.~J. and D.~V. Jones. 2003. An Introduction to the Theory of Point
  Processes: Elementary Theory of Point Processes. Springer.
\bibAnnoteFile{daley2003introduction}

\bibitem[{Dassios and Zhao(2011)}]{dassios2011dynamic}
Dassios, A. and H.~Zhao. 2011. A dynamic contagion process. Advances in applied
  probability 43:814--846.
\bibAnnoteFile{dassios2011dynamic}

\bibitem[{Drummond et~al.(2006{\natexlab{a}})Drummond, Ho, Phillips, and
  Rambaut}]{drummond06}
Drummond, A., S.~Ho, M.~Phillips, and A.~Rambaut. 2006{\natexlab{a}}. Relaxed
  phylogenetics and dating with confidence. PLoS Biology 4:e88.
\bibAnnoteFile{drummond06}

\bibitem[{Drummond et~al.(2006{\natexlab{b}})Drummond, Ho, Phillips, and
  Rambaut}]{drummond2006relaxed}
Drummond, A.~J., S.~Y. Ho, M.~J. Phillips, and A.~Rambaut. 2006{\natexlab{b}}.
  Relaxed phylogenetics and dating with confidence. PLoS Biol 4:e88.
\bibAnnoteFile{drummond2006relaxed}

\bibitem[{Dudas et~al.(2017)Dudas, Carvalho, Bedford, Tatem, Baele, Faria,
  Park, Ladner, Arias, Asogun et~al.}]{dudas2017virus}
Dudas, G., L.~M. Carvalho, T.~Bedford, A.~J. Tatem, G.~Baele, N.~R. Faria,
  D.~J. Park, J.~T. Ladner, A.~Arias, D.~Asogun, et~al. 2017. Virus genomes
  reveal factors that spread and sustained the ebola epidemic. Nature
  544:309--315.
\bibAnnoteFile{dudas2017virus}

\bibitem[{Eddelbuettel and Fran\c{c}ois(2011)}]{eddelbuettel2011rcpp}
Eddelbuettel, D. and R.~Fran\c{c}ois. 2011. {Rcpp}: Seamless {R} and {C++}
  integration. Journal of Statistical Software 40:1--18.
\bibAnnoteFile{eddelbuettel2011rcpp}

\bibitem[{Faria et~al.(2014)Faria, Rambaut, Suchard, Baele, Bedford, Ward,
  Tatem, Sousa, Arinaminpathy, P{\'e}pin et~al.}]{faria2014early}
Faria, N.~R., A.~Rambaut, M.~A. Suchard, G.~Baele, T.~Bedford, M.~J. Ward,
  A.~J. Tatem, J.~D. Sousa, N.~Arinaminpathy, J.~P{\'e}pin, et~al. 2014. The
  early spread and epidemic ignition of hiv-1 in human populations. science
  346:56--61.
\bibAnnoteFile{faria2014early}

\bibitem[{Felsenstein(1978)}]{felsenstein1978number}
Felsenstein, J. 1978. The number of evolutionary trees. Systematic zoology
  27:27--33.
\bibAnnoteFile{felsenstein1978number}

\bibitem[{Felsenstein(1985)}]{felsenstein85}
Felsenstein, J. 1985. Phylogenies and the comparative method. American
  Naturalist 125:1--15.
\bibAnnoteFile{felsenstein85}

\bibitem[{Ferreira and Suchard(2008)}]{ferreira2008bayesian}
Ferreira, M.~A. and M.~A. Suchard. 2008. Bayesian analysis of elapsed times in
  continuous-time markov chains. Canadian Journal of Statistics 36:355--368.
\bibAnnoteFile{ferreira2008bayesian}

\bibitem[{Fisher et~al.(2021)Fisher, Ji, Zhang, Lemey, and
  Suchard}]{fisher2021relaxed}
Fisher, A.~A., X.~Ji, Z.~Zhang, P.~Lemey, and M.~A. Suchard. 2021. Relaxed
  random walks at scale. Systematic Biology .
\bibAnnoteFile{fisher2021relaxed}

\bibitem[{Fox et~al.(2016)Fox, Schoenberg, Gordon et~al.}]{fox2016spatially}
Fox, E.~W., F.~P. Schoenberg, J.~S. Gordon, et~al. 2016. Spatially
  inhomogeneous background rate estimators and uncertainty quantification for
  nonparametric hawkes point process models of earthquake occurrences. The
  Annals of Applied Statistics 10:1725--1756.
\bibAnnoteFile{fox2016spatially}

\bibitem[{Freckleton(2012)}]{freckleton2012fast}
Freckleton, R.~P. 2012. Fast likelihood calculations for comparative analyses.
  Methods in Ecology and Evolution 3:940--947.
\bibAnnoteFile{freckleton2012fast}

\bibitem[{Gill et~al.(2013)Gill, Lemey, Faria, Rambaut, Shapiro, and
  Suchard}]{gill2013improving}
Gill, M.~S., P.~Lemey, N.~R. Faria, A.~Rambaut, B.~Shapiro, and M.~A. Suchard.
  2013. Improving bayesian population dynamics inference: a coalescent-based
  model for multiple loci. Molecular biology and evolution 30:713--724.
\bibAnnoteFile{gill2013improving}

\bibitem[{Gire et~al.(2014)Gire, Goba, Andersen, Sealfon, Park, Kanneh, Jalloh,
  Momoh, Fullah, Dudas et~al.}]{gire2014genomic}
Gire, S.~K., A.~Goba, K.~G. Andersen, R.~S. Sealfon, D.~J. Park, L.~Kanneh,
  S.~Jalloh, M.~Momoh, M.~Fullah, G.~Dudas, et~al. 2014. Genomic surveillance
  elucidates ebola virus origin and transmission during the 2014 outbreak.
  science 345:1369--1372.
\bibAnnoteFile{gire2014genomic}

\bibitem[{Haario et~al.(2001)Haario, Saksman, Tamminen
  et~al.}]{haario2001adaptive}
Haario, H., E.~Saksman, J.~Tamminen, et~al. 2001. An adaptive metropolis
  algorithm. Bernoulli 7:223--242.
\bibAnnoteFile{haario2001adaptive}

\bibitem[{Habbema et~al.(1974)Habbema, JDF, Van~den Broek
  et~al.}]{habbema1974stepwise}
Habbema, J., H.~JDF, K.~Van~den Broek, et~al. 1974. A stepwise discriminant
  analysis program using density estimation. .
\bibAnnoteFile{habbema1974stepwise}

\bibitem[{Hasegawa et~al.(1985)Hasegawa, Kishino, and
  Yano}]{hasegawa1985dating}
Hasegawa, M., H.~Kishino, and T.-a. Yano. 1985. Dating of the human-ape
  splitting by a molecular clock of mitochondrial dna. Journal of molecular
  evolution 22:160--174.
\bibAnnoteFile{hasegawa1985dating}

\bibitem[{Hawkes(1972)}]{hawkes1972spectra}
Hawkes, A. 1972. Spectra of some mutually exciting point processes with
  associated variables. Stochastic point processes Pages~261--271.
\bibAnnoteFile{hawkes1972spectra}

\bibitem[{Hawkes(1973)}]{hawkes1973cluster}
Hawkes, A. 1973. Cluster models for earthquakes-regional comparisons. Bull.
  Int. Stat. Inst. 45:454--461.
\bibAnnoteFile{hawkes1973cluster}

\bibitem[{Hawkes(1971{\natexlab{a}})}]{hawkes1971point}
Hawkes, A.~G. 1971{\natexlab{a}}. Point spectra of some mutually exciting point
  processes. Journal of the Royal Statistical Society: Series B
  (Methodological) 33:438--443.
\bibAnnoteFile{hawkes1971point}

\bibitem[{Hawkes(1971{\natexlab{b}})}]{hawkes1971spectra}
Hawkes, A.~G. 1971{\natexlab{b}}. Spectra of some self-exciting and mutually
  exciting point processes. Biometrika 58:83--90.
\bibAnnoteFile{hawkes1971spectra}

\bibitem[{Hawkes(2018)}]{hawkes2018hawkes}
Hawkes, A.~G. 2018. Hawkes processes and their applications to finance: a
  review. Quantitative Finance 18:193--198.
\bibAnnoteFile{hawkes2018hawkes}

\bibitem[{Ho and An\'e(2014)}]{ho2014linear}
Ho, L. S.~T. and C.~An\'e. 2014. A linear-time algorithm for {G}aussian and
  non-{G}aussian trait evolution models. Systematic Biology 3:397--402.
\bibAnnoteFile{ho2014linear}

\bibitem[{Holbrook et~al.(2020)Holbrook, Lemey, Baele, Dellicour, Brockmann,
  Rambaut, and Suchard}]{holbrook2020massive}
Holbrook, A.~J., P.~Lemey, G.~Baele, S.~Dellicour, D.~Brockmann, A.~Rambaut,
  and M.~A. Suchard. 2020. Massive parallelization boosts big bayesian
  multidimensional scaling. Journal of Computational and Graphical Statistics
  Pages~1--34.
\bibAnnoteFile{holbrook2020massive}

\bibitem[{Holbrook et~al.(2021)Holbrook, Loeffler, Flaxman, and
  Suchard}]{holbrook2021scalable}
Holbrook, A.~J., C.~E. Loeffler, S.~R. Flaxman, and M.~A. Suchard. 2021.
  Scalable bayesian inference for self-excitatory stochastic processes applied
  to big american gunfire data. Statistics and Computing 31:1--15.
\bibAnnoteFile{holbrook2021scalable}

\bibitem[{Kahle and Wickham(2013)}]{kahle2013ggmap}
Kahle, D. and H.~Wickham. 2013. ggmap: Spatial visualization with ggplot2. The
  R journal 5:144--161.
\bibAnnoteFile{kahle2013ggmap}

\bibitem[{Kelly et~al.(2019)Kelly, Park, Harrigan, Hoff, Lee, Wannier, Selo,
  Mossoko, Njoloko, Okitolonda-Wemakoy et~al.}]{kelly2019real}
Kelly, J.~D., J.~Park, R.~J. Harrigan, N.~A. Hoff, S.~D. Lee, R.~Wannier,
  B.~Selo, M.~Mossoko, B.~Njoloko, E.~Okitolonda-Wemakoy, et~al. 2019.
  Real-time predictions of the 2018--2019 ebola virus disease outbreak in the
  democratic republic of the congo using hawkes point process models. Epidemics
  28:100354.
\bibAnnoteFile{kelly2019real}

\bibitem[{Kim(2011)}]{kim2011spatio}
Kim, H. 2011. Spatio-temporal point process models for the spread of avian
  influenza virus (H5N1). Ph.D. thesis UC Berkeley.
\bibAnnoteFile{kim2011spatio}

\bibitem[{Kobayashi and Lambiotte(2016)}]{kobayashi2016tideh}
Kobayashi, R. and R.~Lambiotte. 2016. Tideh: Time-dependent hawkes process for
  predicting retweet dynamics. \emph{in} Proceedings of the International AAAI
  Conference on Web and Social Media vol.~10.
\bibAnnoteFile{kobayashi2016tideh}

\bibitem[{Leimkuhler and Reich(2004)}]{leimkuhler2004simulating}
Leimkuhler, B. and S.~Reich. 2004. Simulating {H}amiltonian dynamics vol.~14.
  Cambridge university press.
\bibAnnoteFile{leimkuhler2004simulating}

\bibitem[{Lemey et~al.(2009)Lemey, Rambaut, Drummond, and
  Suchard}]{lemey2009bayesian}
Lemey, P., A.~Rambaut, A.~Drummond, and M.~Suchard. 2009. Bayesian
  phylogeography finds its roots. PLoS Computational Biology 5:e1000520.
\bibAnnoteFile{lemey2009bayesian}

\bibitem[{Lemey et~al.(2010)Lemey, Rambaut, Welch, and
  Suchard}]{lemey2010phylogeography}
Lemey, P., A.~Rambaut, J.~Welch, and M.~Suchard. 2010. Phylogeography takes a
  relaxed random walk in continuous space and time. Molecular Biology and
  Evolution 27:1877--1885.
\bibAnnoteFile{lemey2010phylogeography}

\bibitem[{Loeffler and Flaxman(2018)}]{loeffler2018gun}
Loeffler, C. and S.~Flaxman. 2018. Is gun violence contagious? a spatiotemporal
  test. Journal of quantitative criminology 34:999--1017.
\bibAnnoteFile{loeffler2018gun}

\bibitem[{Mau et~al.(1999)Mau, Newton, and Larget}]{mau1999bayesian}
Mau, B., M.~A. Newton, and B.~Larget. 1999. Bayesian phylogenetic inference via
  markov chain monte carlo methods. Biometrics 55:1--12.
\bibAnnoteFile{mau1999bayesian}

\bibitem[{Mei and Eisner(2017)}]{mei2017neural}
Mei, H. and J.~M. Eisner. 2017. The neural {H}awkes process: A neurally
  self-modulating multivariate point process. Pages~6754--6764 \emph{in}
  Advances in Neural Information Processing Systems.
\bibAnnoteFile{mei2017neural}

\bibitem[{Meyer et~al.(2014)Meyer, Held et~al.}]{meyer2014power}
Meyer, S., L.~Held, et~al. 2014. Power-law models for infectious disease
  spread. The Annals of Applied Statistics 8:1612--1639.
\bibAnnoteFile{meyer2014power}

\bibitem[{Mohler(2014)}]{mohler2014marked}
Mohler, G. 2014. Marked point process hotspot maps for homicide and gun crime
  prediction in chicago. International Journal of Forecasting 30:491--497.
\bibAnnoteFile{mohler2014marked}

\bibitem[{Mohler et~al.(2013)}]{mohler2013modeling}
Mohler, G. et~al. 2013. Modeling and estimation of multi-source clustering in
  crime and security data. The Annals of Applied Statistics 7:1525--1539.
\bibAnnoteFile{mohler2013modeling}

\bibitem[{Neal(2011)}]{neal2011mcmc}
Neal, R.~M. 2011. {MCMC} using {H}amiltonian dynamics. Handbook of Markov Chain
  Monte Carlo 2.
\bibAnnoteFile{neal2011mcmc}

\bibitem[{Nickel and Le(2020)}]{nickel2020learning}
Nickel, M. and M.~Le. 2020. Learning multivariate hawkes processes at scale.
  arXiv preprint arXiv:2002.12501 .
\bibAnnoteFile{nickel2020learning}

\bibitem[{Ogata(1988)}]{ogata1988statistical}
Ogata, Y. 1988. Statistical models for earthquake occurrences and residual
  analysis for point processes. Journal of the American Statistical association
  83:9--27.
\bibAnnoteFile{ogata1988statistical}

\bibitem[{Park et~al.(2019)Park, Schoenberg, Bertozzi, and
  Brantingham}]{park2019investigating}
Park, J., F.~P. Schoenberg, A.~L. Bertozzi, and P.~J. Brantingham. 2019.
  Investigating clustering and violence interruption in gang-related violent
  crime data using spatial-temporal point processes with covariates .
\bibAnnoteFile{park2019investigating}

\bibitem[{Pearl(1982)}]{pearl1982reverend}
Pearl, J. 1982. Reverend {B}ayes on inference engines: A distributed
  hierarchical approach. Pages~133--136 \emph{in} AAAI-82: Proceedings of the
  Second National Conference on Artificial Intelligence.
\bibAnnoteFile{pearl1982reverend}

\bibitem[{Pybus et~al.(2012)Pybus, Suchard, Lemey, Bernardin, Rambaut,
  Crawford, Gray, Arinaminpathy, Stramer, Busch et~al.}]{pybus2012unifying}
Pybus, O.~G., M.~A. Suchard, P.~Lemey, F.~J. Bernardin, A.~Rambaut, F.~W.
  Crawford, R.~R. Gray, N.~Arinaminpathy, S.~L. Stramer, M.~P. Busch, et~al.
  2012. Unifying the spatial epidemiology and molecular evolution of emerging
  epidemics. Proceedings of the National Academy of Sciences 109:15066--15071.
\bibAnnoteFile{pybus2012unifying}

\bibitem[{Rambaut et~al.(2008)Rambaut, Pybus, Nelson, Viboud, Taubenberger, and
  Holmes}]{rambaut2008genomic}
Rambaut, A., O.~G. Pybus, M.~I. Nelson, C.~Viboud, J.~K. Taubenberger, and
  E.~C. Holmes. 2008. The genomic and epidemiological dynamics of human
  influenza a virus. Nature 453:615--619.
\bibAnnoteFile{rambaut2008genomic}

\bibitem[{Reinhart(2018)}]{reinhart2018review}
Reinhart, A. 2018. A review of self-exciting spatio-temporal point processes
  and their applications. Statistical Science 33:299--318.
\bibAnnoteFile{reinhart2018review}

\bibitem[{Rizoiu et~al.(2017)Rizoiu, Lee, Mishra, and Xie}]{rizoiu2017tutorial}
Rizoiu, M.-A., Y.~Lee, S.~Mishra, and L.~Xie. 2017. A tutorial on hawkes
  processes for events in social media. arXiv preprint arXiv:1708.06401 .
\bibAnnoteFile{rizoiu2017tutorial}

\bibitem[{Rizoiu et~al.(2018)Rizoiu, Mishra, Kong, Carman, and
  Xie}]{rizoiu2018sir}
Rizoiu, M.-A., S.~Mishra, Q.~Kong, M.~Carman, and L.~Xie. 2018. Sir-{H}awkes:
  Linking epidemic models and {H}awkes processes to model diffusions in finite
  populations. Pages~419--428 \emph{in} Proceedings of the 2018 World Wide Web
  Conference on World Wide Web International World Wide Web Conferences
  Steering Committee.
\bibAnnoteFile{rizoiu2018sir}

\bibitem[{Robert(1976)}]{robert1976choice}
Robert, P. 1976. On the choice of smoothing parameters for parzen estimators of
  probability density functions. IEEE Transactions on Computers 25:1175--1179.
\bibAnnoteFile{robert1976choice}

\bibitem[{Ronquist et~al.(2012)Ronquist, Teslenko, Van Der~Mark, Ayres,
  Darling, H{\"o}hna, Larget, Liu, Suchard, and
  Huelsenbeck}]{ronquist2012mrbayes}
Ronquist, F., M.~Teslenko, P.~Van Der~Mark, D.~L. Ayres, A.~Darling,
  S.~H{\"o}hna, B.~Larget, L.~Liu, M.~A. Suchard, and J.~P. Huelsenbeck. 2012.
  Mrbayes 3.2: efficient bayesian phylogenetic inference and model choice
  across a large model space. Systematic biology 61:539--542.
\bibAnnoteFile{ronquist2012mrbayes}

\bibitem[{Salmon et~al.(2011)Salmon, Moraes, Dror, and
  Shaw}]{salmon2011parallel}
Salmon, J.~K., M.~A. Moraes, R.~O. Dror, and D.~E. Shaw. 2011. Parallel random
  numbers: as easy as 1, 2, 3. Pages~1--12 \emph{in} Proceedings of 2011
  International Conference for High Performance Computing, Networking, Storage
  and Analysis.
\bibAnnoteFile{salmon2011parallel}

\bibitem[{Schoenberg(2004)}]{schoenberg2004testing}
Schoenberg, F.~P. 2004. Testing separability in spatial-temporal marked point
  processes. Biometrics Pages~471--481.
\bibAnnoteFile{schoenberg2004testing}

\bibitem[{Schoenberg(2013)}]{schoenberg2013facilitated}
Schoenberg, F.~P. 2013. Facilitated estimation of etas. Bulletin of the
  Seismological Society of America 103:601--605.
\bibAnnoteFile{schoenberg2013facilitated}

\bibitem[{Schoenberg(2020)}]{schoenberg2020nonparametric}
Schoenberg, F.~P. 2020. Nonparametric estimation of variable productivity
  hawkes processes. arXiv preprint arXiv:2003.08858 .
\bibAnnoteFile{schoenberg2020nonparametric}

\bibitem[{Schoenberg et~al.(2019)Schoenberg, Hoffmann, and
  Harrigan}]{schoenberg2019recursive}
Schoenberg, F.~P., M.~Hoffmann, and R.~J. Harrigan. 2019. A recursive point
  process model for infectious diseases. Annals of the Institute of Statistical
  Mathematics 71:1271--1287.
\bibAnnoteFile{schoenberg2019recursive}

\bibitem[{Sinsheimer et~al.(1996)Sinsheimer, Lake, and
  Little}]{sinsheimer1996bayesian}
Sinsheimer, J.~S., J.~A. Lake, and R.~J. Little. 1996. Bayesian hypothesis
  testing of four-taxon topologies using molecular sequence data. Biometrics
  Pages~193--210.
\bibAnnoteFile{sinsheimer1996bayesian}

\bibitem[{Smith et~al.(2009)Smith, Vijaykrishna, Bahl, Lycett, Worobey, Pybus,
  Ma, Cheung, Raghwani, Bhatt et~al.}]{smith2009origins}
Smith, G.~J., D.~Vijaykrishna, J.~Bahl, S.~J. Lycett, M.~Worobey, O.~G. Pybus,
  S.~K. Ma, C.~L. Cheung, J.~Raghwani, S.~Bhatt, et~al. 2009. Origins and
  evolutionary genomics of the 2009 swine-origin h1n1 influenza a epidemic.
  Nature 459:1122--1125.
\bibAnnoteFile{smith2009origins}

\bibitem[{Suchard et~al.(2003)Suchard, Kitchen, Sinsheimer, and
  Weiss}]{Suchard03HPM}
Suchard, M., C.~Kitchen, J.~Sinsheimer, and R.~Weiss. 2003. Hierarchical
  phylogenetic models for analyzing multipartite sequence data. Systematic
  Biology 52:649--664.
\bibAnnoteFile{Suchard03HPM}

\bibitem[{Suchard and Rambaut(2009)}]{suchard2009many}
Suchard, M. and A.~Rambaut. 2009. Many-core algorithms for statistical
  phylogenetics. Bioinformatics 25:1370--1376.
\bibAnnoteFile{suchard2009many}

\bibitem[{Suchard et~al.(2001)Suchard, Weiss, and Sinsheimer}]{suchard01}
Suchard, M., R.~Weiss, and J.~Sinsheimer. 2001. Bayesian selection of
  continuous-time {M}arkov chain evolutionary models. Molecular Biology and
  Evolution 18:1001--1013.
\bibAnnoteFile{suchard01}

\bibitem[{Suchard et~al.(2018)Suchard, Lemey, Baele, Ayres, Drummond, and
  Rambaut}]{suchard2018bayesian}
Suchard, M.~A., P.~Lemey, G.~Baele, D.~L. Ayres, A.~J. Drummond, and
  A.~Rambaut. 2018. Bayesian phylogenetic and phylodynamic data integration
  using {BEAST} 1.10. Virus Evolution 4:vey016.
\bibAnnoteFile{suchard2018bayesian}

\bibitem[{Wadman(2021)}]{wadman2021united}
Wadman, M. 2021. United states rushes to fill void in viral sequencing. Science
  (New York, NY) 371:657--658.
\bibAnnoteFile{wadman2021united}

\bibitem[{Wickham(2016)}]{ggplot}
Wickham, H. 2016. ggplot2: Elegant Graphics for Data Analysis. Springer-Verlag
  New York.
\bibAnnoteFile{ggplot}

\bibitem[{{World Health Organization}(2015)}]{world2015ebola}
{World Health Organization}. 2015. Who: Ebola situation report 30 december 2015
  .
\bibAnnoteFile{world2015ebola}

\bibitem[{Yang(1994)}]{yang1994maximum}
Yang, Z. 1994. Maximum likelihood phylogenetic estimation from dna sequences
  with variable rates over sites: approximate methods. Journal of Molecular
  evolution 39:306--314.
\bibAnnoteFile{yang1994maximum}

\bibitem[{Yang and Rannala(1997)}]{yang1997bayesian}
Yang, Z. and B.~Rannala. 1997. Bayesian phylogenetic inference using dna
  sequences: a markov chain monte carlo method. Molecular biology and evolution
  14:717--724.
\bibAnnoteFile{yang1997bayesian}

\bibitem[{Yuan et~al.(2021)Yuan, Schoenberg, and Bertozzi}]{Yuan2021FastEO}
Yuan, B., F.~Schoenberg, and A.~Bertozzi. 2021. Fast estimation of multivariate
  spatiotemporal hawkes processes and network reconstruction. Annals of the
  Institute of Statistical Mathematics Pages~1--26.
\bibAnnoteFile{Yuan2021FastEO}

\bibitem[{Zhang et~al.(2020)Zhang, Lipani, Kirnap, and Yilmaz}]{zhang2020self}
Zhang, Q., A.~Lipani, O.~Kirnap, and E.~Yilmaz. 2020. Self-attentive hawkes
  process. Pages~11183--11193 \emph{in} International Conference on Machine
  Learning PMLR.
\bibAnnoteFile{zhang2020self}

\bibitem[{Zhuang et~al.(2004)Zhuang, Ogata, and
  Vere-Jones}]{zhuang2004analyzing}
Zhuang, J., Y.~Ogata, and D.~Vere-Jones. 2004. Analyzing earthquake clustering
  features by using stochastic reconstruction. Journal of Geophysical Research:
  Solid Earth 109.
\bibAnnoteFile{zhuang2004analyzing}

\bibitem[{Zuo et~al.(2020)Zuo, Jiang, Li, Zhao, and Zha}]{zuo2020transformer}
Zuo, S., H.~Jiang, Z.~Li, T.~Zhao, and H.~Zha. 2020. Transformer hawkes
  process. Pages~11692--11702 \emph{in} International Conference on Machine
  Learning PMLR.
\bibAnnoteFile{zuo2020transformer}

\end{thebibliography}

\end{document}